\documentclass[3p]{elsarticle}

\usepackage{hyperref}

\journal{arXiv}

\usepackage{amsfonts}
\usepackage{amsmath, amscd, amssymb, bm, amsbsy, epsf, amsthm}
    \usepackage{enumerate}
    \usepackage{paralist}
\usepackage{graphicx,color}
\usepackage{subfigure}
\usepackage{mathrsfs}
\usepackage{verbatim}
\usepackage{inputenc}
\usepackage{diagbox}
\usepackage{algorithm}
\usepackage{algorithmicx}
\usepackage{algpseudocode}
\usepackage{hyperref}
\usepackage{pifont}
\usepackage{inputenc}
\usepackage{bbm, dsfont}
\usepackage{multicol}
\usepackage{balance}
\usepackage{verbatim, booktabs}
\usepackage{multirow, url}
\usepackage[table]{xcolor}
\usepackage{diagbox}
\usepackage{tikz}

\usepackage{cmbright}
\usepackage{setspace}

\newtheorem{theorem}{Theorem}
\newtheorem{lemma}[theorem]{Lemma}
\newtheorem{corollary}[theorem]{Corollary}
\newtheorem{proposition}[theorem]{Proposition}
\newdefinition{definition}{Definition}
\newdefinition{hypothesis}{Hypothesis}
\newdefinition{problem}{Problem}
\newdefinition{remark}{Remark}
\newdefinition{example}{Example}
\def\N{\boldsymbol{\mathbbm{N}}}
\def\R{\boldsymbol{\mathbbm{R}}}
\def\defining{\overset{\mathbf{def}}=}
\def\P{\mathcal{P}}

\def\O{\mathcal{O}}

\def\G{\mathcal{G}}

\def\Exp{\boldsymbol{\mathbb{E}}}
\def\prob{\boldsymbol{\mathbbm{P} } }
\def\Var{\mathbb{V}\mathrm{ar}}

\def\pt{p^{T}}
\def\x{\mathbf{x} }
\def\y{\mathbf{y} }
\def\z{\mathbf{z} }
\def\T{\mathcal{T}}
\def\p{\mathbf{p}}
\def\q{\mathbf{q}}
\def\r{\mathbf{r}}

\def\zero{\mathbf{0} }

\DeclareMathOperator*{\troot}{rt}
\DeclareMathOperator*{\mrf}{\Omega_{MRF}}
\DeclareMathOperator*{\irf}{\Omega_{IRF}}
\DeclareMathOperator*{\dist}{dist}

\DeclareMathOperator*{\inv}{inv}

\newcommand\lt{\mathop{}\!\mathrm{left} } %\ind{\mathbb{I} }
\newcommand\rt{\mathop{}\!\mathrm{right} } %\ind{\mathbb{I} }
\newcommand\X{ \mathsf{X} } %{\mathbb{E}}
\newcommand\Y{ \mathsf{Y} } %{\boldsymbol{\mathbb{P}}}
\newcommand\Z{ \mathsf{Z} } %\ind{\mathbb{I} }
\newcommand\uph{ \mathsf{H} } %\ind{\mathbb{I} }
\def\eversor{\widehat{\mathbf{e} } }

\begin{document}
%\doublespacing
\begin{frontmatter}

\title{A Mathematical Assessment of the 
Isolation Random \\ Forest 
Method for Anomaly Detection in Big Data}
\tnotetext[mytitlenote]{This material is based upon work supported by project HERMES 45713 from Universidad Nacional de Colombia,
Sede Medell\'in.}

\author[mymainaddress]{Fernando A Morales} %\fnref{myfootnote} }
\cortext[mycorrespondingauthor]{Corresponding Author}
\ead{famoralesj@unal.edu.co}

\author[mysecondaryaddress]{Jorge M Ram\'irez}

\author[mymainaddress]{Edgar A Ramos}
%}
%\address{Escuela de Matem\'aticas
%Universidad Nacional de Colombia, Sede Medell\'in \\
%Calle 59 A No 63-20 - Bloque 43, of 106,
%Medell\'in - Colombia}

%\author[mymainaddress]{Fernando A Morales}
%\ead[url]{www.unal.edu.co}

\address[mymainaddress]{Escuela de Matem\'aticas
Universidad Nacional de Colombia, Sede Medell\'in \\
Carrera 65 \# 59A--110, Bloque 43, of 106,
Medell\'in - Colombia}

\address[mysecondaryaddress]{Computer Science and Mathematics Division, Oak Ridge National Laboratory, TN, USA}

%\address[mymainaddress]{Department of Mathematics, Kidder Hall 368, Oregon State University. Corvallis, OR 97331-4605}
%\address[mysecondaryaddress]{Grupo \'Exito, Medell\'in-Colombia}

\begin{abstract}
  We present the mathematical analysis of the Isolation Random Forest Method (IRF Method) for anomaly detection, introduced in \cite{LiuIRF} and \cite{Liu2012IsolationBasedAD}. We prove that the IRF space can be endowed with a probability induced by the Isolation Tree algorithm (iTree). In this setting, the convergence of the IRF method is proved, using the Law of Large Numbers. A couple of counterexamples are presented to show that the method is inconclusive and no certificate of quality can be given, when using it as a means to detect anomalies. Hence, an alternative version of the method is proposed whose mathematical foundation is fully justified. Furthermore, a criterion for choosing the number of sampled trees needed to guarantee confidence intervals of the numerical results is presented. Finally, numerical experiments are presented to compare the performance of the classic method with the proposed one.
\end{abstract}


\end{frontmatter}

%\linenumbers
%%%%%%%%%%%%%%%%%%%%%%%%%%%%%%%%%%%%%%%%%%%%%%%%%%%%%%%%%%%%%%%%%%

%%%%%%%%%%%%%%%% appears after author names %%%
%\dedicatory{The first author wishes to express his gratitude towards Professor Ing. Hugo Lara (in memoriam) for his influence in the understanding of mechanics and fluid dynamics.}
%%%%%%%%%%%%%%%%%%%%%%%%%%%%%%%%%%%%%%%%%%%%%%%%%%%%%%%

%%%%%%%%%%%%%%%%%%%%%%%%%%%%%%%%%%%%%%%%%%%%%%%%%%%%%%%
%
%\maketitle
%
%\tableofcontents
%%%%%%%%%%%%%%%%%%%%%%%%%%%%%%%%%%%%%%%%%%%%%%%%%%%%%%%%%%%%
%%%%%%%%%%%%%%%%%%%%%%%%%%%%%%%%%%%%%%%%%%%%%%%%%%%%%%%%%%%%
%\newpage  %% This command forces a page break
%\large %%% LARGE for PROOFREADING ONLY %%%
%
%
%
%%%%%%%%%%%%%%%%%%%%%%%%%%%%%%%%%%%%%%%%%%%%%%%%%%%%%%%%%%%%%%%%%%%%%%%%%%%%%%%%%%

\section{Introduction}\label{Sec Introduction}
%%%%%%%%%%%%%%%%%%%%%%%%%%%%%%%%%%%%%%%%%%%%%%%%%%%%%%%%%%%%%%%%%%%%%%%%%%%%%%%%%%%%%%%%%%
%
%
Anomaly detection in big data is an important field of research due to its applications, the presence of anomalies may indicate disease of individuals, fraudulent transactions and network security breaches, among others. There is a remarkable number of methods for anomaly detection following different paradigms, some of these are  distance-based (see \cite{AngiulliPizzuti, BaySchwabacher, KnorrNg}), classification-based (see \cite{AbeZadronyLangford, ShiHorvath}), cluster-based (see \cite{HeXuXiaofeiDeng}), density-based (see \cite{CriminisiShotonKonukoglu}, \cite{RamGray}) and isolation-based (see \cite{LiuIRF, Liu2012IsolationBasedAD, BandaragodaIsolation, BandaragodaPaper}). In general, it is not possible to compare these methods from a unified point of view, because many of them were developed for particular types problems with data sets satisfying specific hypotheses, although some techniques are applicable to a broader class of problems, see \cite{ThudumuSurvey} and \cite{ChandolaSurvey} for comprehensive surveys on the field. 

In the present work we focus on the mathematical analysis of the Isolation Random Forest Method, from now on, denoted by IRF. Despite the popularity of the IRF method, to the authors' best knowledge, it has not been analyzed mathematically.
%no mathematical analysis has been done to it. 
For instance, there is no rigorous proof that the method converges, there is no analysis about the number of iterations needed to guarantee confidence intervals for the computed values. Some scenarios where the method performs poorly have been pointed out in \cite{LiuIRF} and \cite{Liu2012IsolationBasedAD}, but there are no general recommendations/guidelines for a setting where the IRF Method runs successfully. In the present work, all these aspects are addressed with mathematical rigor. 

Due to its popularity, the IRF method has been analyzed from the empirical point of view and some of its deficiencies, detected experimentally, have been addressed and an enhanced with good results, from the empirical point of view. Following the Machine Learning modern point of view, most of these attempts, consider that one or more samples must be taken from the available data as a training set. The ultimate goal of such training is to build a probabilistic distribution locating the outliers according to a probabilistic distribution on the $ \R^{d} $ region the data containing the data; once the distribution is deduced the non-used data serve as testing instances. Some of the IRF method's aforementioned weaknesses were that it failed to detect clustered outliers (see \cite{LiuTingZhouSciForest}) and that for certain particular data configurations some ghost sub-regions in the probability distribution show up (see \cite{HaririKindBrunner}). In both works the authors coincide in modifying the data split criterion in which the isolation trees are constructed, as a way to amend the method's flaws. To that end, in \cite{LiuTingZhouSciForest} the authors propose a separation hyperplane for data splitting, which is computed using an empirical optimization process executed on a set of random oblique hyperplanes generated in every iteration of the algorithm presented algorithm (SciForest: Isolation Forest with Split-selection Criterion), the idea is similar to that in \cite{OlbiqueDecisionTrees}. This way, the randomness is preserved, while refined through a deterministic method of optimization. In the reference \cite{HaririKindBrunner}, the authors tackle the original IRF method's limitations with two approaches. In both cases the aim is to use random oblique hyerplanes instead of exclusively using axis-parallel hypeplanes as in the original method; unlike \cite{LiuTingZhouSciForest}, no optimization process is applied. Both works report that the novelty proposed methods outperform the original one. However, their results are established empirically in the same way that the flaws of the original IRF method were detected. Essentially, the mainstream is to keep enhancing the original method using experimental validation. We differ radically from this approach as we seek to understand the method through mathematical arguments of theoretical nature; later we use experiments to illustrate our results.

The paper is organized as follows, in the introductory section the notation and general setting are presented; the IRF Method is reviewed for the sake of completeness and a proof is given that the iTree algorithm is well-defined. \textsc{Section} \ref{Sec General Case} presents the analysis of the method in the general setting: the underlying probabilistic structure is stated, the convergence of the method is established, the cardinality of the isolation random forest (IRF) is presented and two examples are given to analytically prove the inconclusiveness of the IRF Method in its original form. In \textsc{Section} \ref{Sec 1D Setting}, IRF is analyzed for the 1D case and proved to be a suitable tool for anomaly detection in this particular setting. This is done recalling a closely related algorithm, next, estimates for the values of the expected height and variance are given; the section closes introducing an alternative version of method (the Direction Isolation Random Forest Method, DIRF) whose mathematical foundation is fully justified. Finally, \textsc{Section} \ref{Sec Numerical Experiments} presents numerical examples examining the performance of both methods IRF and DIRF.    
\begin{remark}[A note about the paper's organization]
	The Authors realize that the organization of the paper may rise some questions and disagreements, as it first visits the multidimensional case $ \R^{d} $ (\textsc{Section} \ref{Sec General Case}) and then analyzes the 1D case (\textsc{Section} \ref{Sec 1D Setting}). At first sight it may seem natural to exchange the order of these sections due to the level of generality they expose. The Authors did this exercise (privately) and concluded that such exposition would enlighten new features, but obscure some aspects they consider as priority in this work. The actual choice was made pursuing a balance between brevity, clarity and avoiding redundancy in the exposition. However, for a reader more interested in the mathematical techniques, than the analysis of the algorithms, it may be more beneficial to read first section \ref{Sec 1D Setting} and then section \ref{Sec General Case}. 
\end{remark}
%
%
%
%%%%%%%%%%%%%%%%%%%%%%%%%%%%%%%%%%%%%%%%%%%%%%%%%%%%%%%%%%%%%%%%%%%%%%%%%%%%%%%%%%
\subsection{Preliminaries}\label{Sec Preliminaries of the IRF Method}
%%%%%%%%%%%%%%%%%%%%%%%%%%%%%%%%%%%%%%%%%%%%%%%%%%%%%%%%%%%%%%%%%%%%%%%%%%%%%%%%%%%%%%%%%%
%
%
In this section the general setting and preliminaries of the problem are presented. We start introducing the mathematical notation. For any natural number $ N\in \N $, the symbol $ [N] \defining \{ 1, 2, \ldots , N \} $ indicates the set/window of the first $ N $ natural numbers. For any set $ E $ we denote by $ \# E  $, $ \wp(E) $ its cardinal and power set respectively. For any interval $ I \subseteq \R $ we denote by $ \vert I \vert $ its length. 
%A particularly important set is $ \mathcal{S}_{N} $, the set of all permutations in $ [N] $, its elements will be usually denoted by the Greek letters $ \sigma, \tau $, etc.  
Random variables will be represented with upright capital letters, namely $ \X, \Y, \Z $ and its respective expectations with $ \Exp(\X),  \Exp(\Y), \Exp(\Z)$. Vectors are indicated with bold letters, namely $ \p, \q, \r ... $, etc. The canonical basis in $ \R^{d} $ is written $ \{ \eversor_{1}, \ldots, \eversor_{d}\} $, the projections from $ \R^{d} $ onto the $ j $-th coordinate are denoted by $ \pi_{j}(\x) \defining \x\cdot \eversor_{j} $ for all $ j \in [d] $, where $ \x \in \R^{d} $. 
%Particularly important collections of objects will be written with calligraphic characters, e.g. $ \mathcal{A}, \mathcal{D}, \mathcal{E} $ to add emphasis. 

The isolation random tree algorithm (iTree) for a set of points in $ \R^{d} $ is defined recursively as follows  
\begin{definition}[The iTree Algorithm, introduced in \cite{LiuIRF, Liu2012IsolationBasedAD}]\label{Def Isolation Random Tree}
	Let $ S \defining \{\x_{0}, \x_{1}, \ldots, \x_{N} \} $ be a set of points in $ \R^{d} $. 
	\begin{enumerate}[(i)]
		\item An \textbf{isolation random tree} $ T $ (iTree), associated to this set is defined recursively as follows:
		\begin{enumerate}[a.]
			\item Define the tree root as $ \troot(T) \defining S  $.
			
			\item Define the sets 
			\begin{align}\label{Def Random Split Coordinate}
			& \pi_{j} (S) \defining \{ \pi_{j}(\x):\x \in S\} , & 
			& 1 \leq j \leq d, &
			& \Omega_{C} \defining \{ j \in [d] : \#\pi_{j}(S) \geq 2  \} .
			\end{align}
			%
			%$ S_{j} \defining \{ \x_{i}\cdot \eversor_{j}: 0 \leq i \leq N \} $, for $ 1 \leq j \leq d $ and $ C \defining \{ j \in [d] : \#S_{j} \geq 2  \} $.
			
			\item If the set $ S $ has two or more points (equivalently, if $ \Omega_{C} \neq \emptyset $), choose randomly $ j $ in $ \Omega_{C} $ and next choose randomly $ p \in ( \min \pi_{j} (S), \max \pi_{j} (S) ) $ (the split value).
			
			\item Perform an isolation random tree on the \textbf{left set} of data $ S_{\lt} \defining \{ \x \in S: \x \cdot \eversor_{j} < p \} $, denoted by $ T_{\lt}$. Next, include the arc $ \big(\troot(T), \troot(T_{\lt})\big) $ in the edges of the tree $ E(T) $; where $ \troot(T_{\lt}) = S_{\lt} $ indicates the root of $ T_{\lt} $.
			
			\item Perform an isolation random tree on the \textbf{right set} of data $ S_{\rt} \defining \{ \x \in S: \x \cdot \eversor_{j} \geq p \} $, denoted by $ T_{\rt}$. Next, define the arc $ \big(\troot(T), \troot(T_{\rt})\big) $ in the edges of the tree $ E(T) $, where $ \troot(T_{\rt}) = S_{\rt} $ indicates the root of $ T_{\rt} $.
		\end{enumerate}
		\item We denote the set of all possible isolation random trees associated to the set $ S $ by $ \irf(S) $, whenever the context is clear, we simply write $ \irf $ and we refer to it as the \textbf{isolation random forest}.
	\end{enumerate}
	For the sake of completeness we present the iTree Algorithm's pseudocode in \ref{Alg iTree Algorithm}.
\end{definition}
%
%
%\begin{remark}\label{Rem Isolation Random Forest Algorithm}
%
%For the sake of completeness we write down the Isolation Random Tree (iTree) Algorithm 

%
\begin{algorithm} %[H]
	\caption{Isolation Random Tree (introduced in \cite{LiuIRF, Liu2012IsolationBasedAD}), returns a rooted tree with $ N $ vertices. The vertices are subsets of the input data set
		$  S = \big\{\x_{n}: n\in [N] \big\} \subseteq \R^{d}$.}
	\label{Alg iTree Algorithm}
	\begin{algorithmic}[1]
		\Procedure{Isolation Random Tree}{Data Set $ S = \big\{ \x_{n}: n\in [N] \big\} $.}
		\State $ \text{VBT} = \emptyset $ \Comment{Initializing the Vertex set in the Binary Tree as empty} 
		\State $ \text{ABT} = \emptyset $ \Comment{Initializing the list of Arcs in the Binary Tree as empty} 
		\State $ \troot = S $ \Comment{Initializing the root as empty} 
		\State $ S \rightarrow \text{VBT} $ \Comment{Push the set $ S $ to \text{VBT} } 
		\Function{Branch}{$ S, \text{VBT}, \text{ABT}, \troot $}
		\If {$ \#\Omega_{C} = 1 $ } \Comment{Checking when to stop}
		%		\If {$ \troot  \neq  \emptyset $}
		%\State $ \troot \defining I_{1} $
		%\Else
		%		\State $ (\troot, S) \rightarrow \text{ABT} $ \Comment{Push the arc $ (\troot, S) $ to \text{ABT} } 
		%		\EndIf
		%\State \Return ($ \text{ABT}, \text{root} $)
		\State \Return $ \text{VBT}, \text{ABT}$
		\Else
		\State choose $ j \in \Omega_{C} $ randomly \Comment{Choosing the split direction.}
		\State define $ \pi_{j} (S) \defining \big\{ \pi_{j}(\x): \x \in S \big\} $, \Comment{Data set projected onto the $ j $-th direction}
		\State choose $ p \in \big( \min \pi_{j} (S), \max \pi_{j} (S) \big) $ randomly \Comment{Choosing the split value.}
		\State define $ S_{\lt} \defining \big\{ \x \in S: \x \cdot \eversor_{j} < p \big\} $, 
		$ S_{\rt} \defining \big\{ \x \in S: \x \cdot \eversor_{j} \geq p \big\} $ \Comment{Left and Right subsets}
		\State $ S_{\lt} \rightarrow \text{VBT} $, $ (\troot, S_{\lt}) \rightarrow \text{ABT} $ \Comment{Push the set $ S_{\lt} $ to \text{VBT} and arc $ (\troot, S_{\lt}) $ to \text{ABT} }
		%		\State $ (\troot, S_{\lt}) \rightarrow \text{ABT} $ \Comment{Push the arc $ (\troot, S_{\lt}) $ to \text{ABT} } 
		\State \textsc{Branch}$ \big(S_{\lt}, \text{VBT}, \text{ABT}, \troot \big) $
		\State $ S_{\rt} \rightarrow \text{VBT} $, $ (\troot, S_{\rt}) \rightarrow \text{ABT} $ \Comment{Push the set $ S_{\rt} $ to \text{VBT} and arc $ (\troot, S_{\rt}) $ to \text{ABT}}
		%		\State $ (\troot, S_{\rt}) \rightarrow \text{ABT} $ \Comment{Push the arc $ (\troot, S_{\rt}) $ to \text{ABT} } 
		\State \textsc{Branch}$ \big(S_{\rt}, \text{VBT}, \text{ABT}, \troot \big) $
		%\EndFunction
		\EndIf
		\EndFunction
		\EndProcedure
	\end{algorithmic}
\end{algorithm}
%
%\end{remark}

%
\begin{proposition}\label{Thm Well posedness of iTree}
	Given an arbitrary set $ S \defining \{\x_{0}, \x_{1}, \ldots, \x_{N} \} $ in $ \R^{d} $, the iTree algorithm described in \textsc{Definition} \ref{Def Isolation Random Tree} needs $ N $ instances to isolate every point in $ S $.
\end{proposition}
\begin{proof}
	We proceed by induction on the number of data $ N $. For $ N = 1 $ the result is trivial. 
	%For $ N = 2 $ the result is direct given that two distinct points $ \x_{0} $ and $ \x_{1} $ must differ in at least one coordinate, namely $ j \in [d] $, moreover, the interval $ (\min \{ \x_{0} \cdot\eversor_{j}, \x_{1} \cdot\eversor_{j}  \}, \max \{ \x_{0} \cdot\eversor_{j}, \x_{1} \cdot\eversor_{j}  \} )  $ is nonempty. Therefore, any hyperplane $ H \defining \{\x\in \R^{d} : \x\cdot \eversor_{j}  = p  \} $, with $ p \in (\min \{ \x_{0} \cdot\eversor_{j}, \x_{1} \cdot\eversor_{j}  \}, \max \{ \x_{0} \cdot\eversor_{j}, \x_{1} \cdot\eversor_{j}  \} ) $ defines $ T_{\lt} = \{\x_{j}: \x_{j} = \min \{ \x_{0} \cdot\eversor_{j}, \x_{1} \cdot\eversor_{j}  \} \} $ and $ T_{\rt} = \{ \x_{j}: \x\cdot \eversor_{j} = \max \{ \x_{0} \cdot\eversor_{j}, \x_{1} \cdot\eversor_{j}  \} \}$, which are both singleton trees. Hence, the algorithm stops after one instance.
	
	Assume now that the result holds true for $ k \leq N - 1 $ and let $ S \defining \{\x_{0}, \x_{1}, \ldots, \x_{N} \} $ be arbitrary in $ \R^{d} $. Since $ \#S \geq 2 $, the set $ \Omega_{C} $ (defined in \eqref{Def Random Split Coordinate}) must be nonempty. Choose randomly an index $ j $ in $ \Omega_{C} $ and choose randomly $ p\in ( \min S_{j}, \max S_{j} ) $. Thus, after one instance of the algorithm the left and right subsets are defined and they satisfy 
	\begin{align*}
	& \#S_{\lt} < N , & 
	& \#S_{\rt} < N , & 
	& \#S_{\lt} + \# S_{\rt} = N .
	\end{align*}
	%
	%	$ \#S_{\lt} < N $, $ \#S_{\rt} < N $, $ \#S_{\lt} + \# S_{\rt} = N $. 
	Then, applying the induction hypothesis on each, the left and right subsets, it follows that the total of needed instances is
	\begin{equation*}
	1 +
	(\#S_{\lt} - 1) + 
	(\#S_{\rt} - 1)
	= \#S_{\lt} + \#S_{\rt} - 1 
	= N - 1 , 
	\end{equation*}
	which completes the proof.
	%\qed
\end{proof}
%
%
%	%
%\end{definition}
\begin{definition}[The IRF Method, introduced in \cite{LiuIRF, Liu2012IsolationBasedAD}]
	\label{Def Isolation Random Forest Method}
	Given an input data set $ S \defining \{\x_{1}, \ldots, \x_{N} \} \subseteq  \R^{d} $, a number of Bernoulli trials $ K $ and an anomaly threshold criterion. 
	\begin{enumerate}[(i)]
		%
		%		\item Find the principal directions of the set $ S $.
		%		
		%		\item Project the data on each direction, i.e., generate  $ \pi_{j}(S) \defining \{ \x \cdot \eversor_{j}: \x \in S\} $, for $ j = 1, \ldots, d $.
		
		\item For each Bernoulli trial $ k = 1, \ldots, K$, perform iTree on $ S $ (see \textsc{Defintion} \ref{Def Isolation Random Tree} and/or \textsc{Algorithm} \ref{Alg iTree Algorithm}) and store the heights $ \{h_{k}(\x): \x\in S\} $ in a Log.
		
		%		\item For each $ S_{i} $ perform iTree $ K $-times and store the heights in separate files, one for each direction. Denote these files by $ \text{Log}_{i} $, $ i = 1, \ldots, d $.
		
		%		\item Compute the average height of $ \x $ for all $ (\x, i) \in S $. Denote this averages by $ \{\upm(\x\cdot \eversor_{i}) : x \in S, i \in [d]\} $.
		
		\item For each $ \x \in S $, define $ \uph_{\text{IRF}}(\x) $ as the average height of the collection of heights $ \{h_{k}(\x): k = 1, \ldots, K\} $.
		
		\item Declare as anomalies 
		$ A \defining \{\x \in S: \uph_{\text{IRF}}(\x) \text{ satisfies the anomaly threshold criterion}  \} $.
	\end{enumerate}
\end{definition}
Observe that due to \textsc{Proposition} \ref{Thm Well posedness of iTree} the iTree Algorithm and consequently the IRF method are well-defined. For brevity, we postpone until \textsc{Theorem} \ref{Thm IRF is a probability space} in \textsc{Section} \ref{Sec General Case}, the proof showing that the isolation random forest can be endowed with a probability measure, defined by the iTree algorithm.
\section{The General Setting}\label{Sec General Case}
%%%%%%%%%%%%%%%%%%%%%%%%%%%%%%%%%%%%%%%%%%%%%%%%%%%%%%%%%%%%%%%%%%%%%%%%%%%%%%%%%%%%%%%%%%
%
%
%
%
This section presents the features of the isolation random forest that can be proved in general, these are: its probability structure, its cardinality and the fact that the IRF method converges and is well-defined. For the analysis of the general setting first we need to introduce a hypothesis
\begin{hypothesis}\label{Hyp Non Repeated Coordinates}
	Given a set of data $ S = \{\x_{0}, \ldots, \x_{N} \} \subseteq \R^{d} $, from now on it will be assumed that no coordinates are repeated i.e.
	\begin{align}\label{Eq Non Repeated Coordinates}
	& \#\pi_{j}(S) = N + 1 , & 
	& \text{for all } j = 1, \ldots, d .
	\end{align}
	Here $ \pi_{j}(S) $ is the $ j $-th projection of the set $ S $ as defined in Equation \eqref{Def Random Split Coordinate}.	
\end{hypothesis}
%
%\begin{comment}
%\color{blue}
%
\begin{definition}\label{Def Projection Intervals}
	Let $ S = \{\x_{0}, \ldots, \x_{N} \} \subseteq \R^{d} $ be a data set satisfying Hypothesis \ref{Hyp Non Repeated Coordinates}. 
	\begin{enumerate}[(i)]
		\item For each $ j \in [d] $ denote by $ \P^{(j)} \defining \big\{ I^{(j)}_{n}: n \in [N] \big\} $, the family of intervals defined by sorting the points of the set $ \pi_{j}(S) = \big\{\x_{n}\cdot \eversor_{j}: n = 0, \ldots, N \big\} $.
		
		%\color{red}
		\item Define the \textbf{grid} of the set by $ G_{S} \defining \prod\limits_{j \, = \, 1}^{d} \pi_{j}(S) $.
		
	\end{enumerate}
	(See \textsc{Figure} \ref{Fig Potential Ancestors} (a) for an illustration when $ S \subseteq \R^{2} $.)
\end{definition}
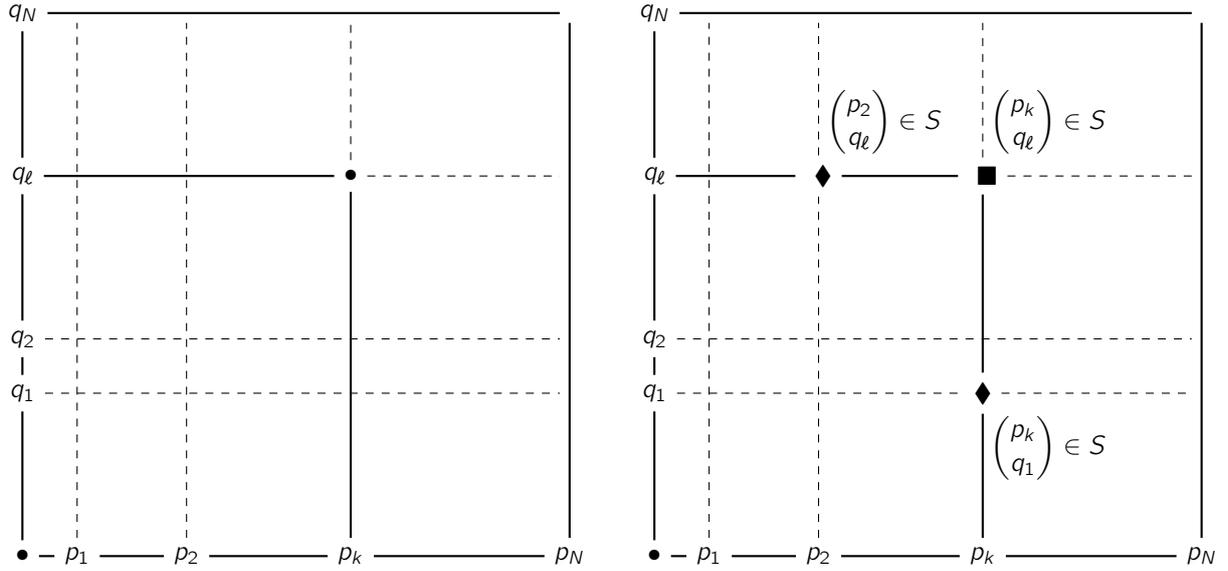
\begin{figure}[h]
	\centering
	\begin{subfigure}[Grid $ G_{S} $, $ S \subseteq \R^{2} $ satisfies Hypothesis \ref{Hyp 2D Configuration}. A potential ancestor is delimited in thick line]
		{
			\begin{tikzpicture}
			[scale=.720,auto=left,every node/.style={}]
			\node (n1) at (0,0)  {$ \bullet $};
			%  \node (n2) at (0,3)  {$q(1)$};
			%  \node (n3) at (0,4)  {$q(2)$};
			%  \node (n4) at (0,7) {$q(\ell)$};
			%  \node(n5) at (0,10) {$q(N)$};
			%  \node(n6) at (1, 0) {$p(1)$};
			%  \node(n7) at (3, 0) {$p(2)$};
			%  \node(n8) at (6, 0) {$p(k)$};
			%  \node(n9) at (10, 0) {$p(N)$};
			\node (n2) at (0,3)  {$q_{1}$};
			\node (n3) at (0,4)  {$q_{2}$};
			\node (n4) at (0,7) {$q_{\ell}$};
			\node(n5) at (0,10) {$q_{N}$};
			\node(n6) at (1, 0) {$p_{1}$};
			\node(n7) at (3, 0) {$p_{2}$};
			\node(n8) at (6, 0) {$p_{k}$};
			\node(n9) at (10, 0) {$p_{N}$};
			\node (n10) at (10,3)  {};
			\node (n11) at (10,4)  {};
			\node (n12) at (10,7) {};
			\node(n13) at (10,10) {};
			\node(n14) at (1, 10) {};
			\node(n15) at (3, 10) {};
			\node(n16) at (6, 10) {};
			\node(n17) at (10, 10) {};
			%  \node(n18) at (6,7){$ \begin{pmatrix}
			%  p(k) \\ q(\ell)
			%  \end{pmatrix} \in \I(S) $};
			\node(n18) at (6,7){$ \bullet $
				%				$ \begin{pmatrix}
				%				p_{k} \\ q_{\ell}
				%				\end{pmatrix} $
			};

			\foreach \from/\to in {n1/n2, n2/n3, n3/n4, n4/n5, 
				n1/n6, n6/n7, n7/n8, n8/n9,
				n4/n18, n5/n17,
				n8/n18, n9/n17}
			\draw[line width = 1.0] (\from) -- (\to);  
			
			\foreach \from/\to  in{n18/n16, n18/n12, n6/n14, n7/n15,
				n2/n10, n3/n11}
			\draw[dashed, line width = 0.5] (\from) -- (\to);
			
			%  \draw (0,0) -- (4,4);
			\end{tikzpicture}
		}
	\end{subfigure}
	\begin{subfigure}[Two possible alternatives (marked with $ \blacksquare $ and $\blacklozenge $) generating the potential ancestor drawn in the figure of the left. ]
		{
			\begin{tikzpicture}
			[scale=.720,auto=left,every node/.style={}]
			\node (n1) at (0,0)  {$ \bullet $};
			%  \node (n2) at (0,3)  {$q(1)$};
			%  \node (n3) at (0,4)  {$q(2)$};
			%  \node (n4) at (0,7) {$q(\ell)$};
			%  \node(n5) at (0,10) {$q(N)$};
			%  \node(n6) at (1, 0) {$p(1)$};
			%  \node(n7) at (3, 0) {$p(2)$};
			%  \node(n8) at (6, 0) {$p(k)$};
			%  \node(n9) at (10, 0) {$p(N)$};
			\node (n2) at (0,3)  {$q_{1}$};
			\node (n3) at (0,4)  {$q_{2}$};
			\node (n4) at (0,7) {$q_{\ell}$};
			\node(n5) at (0,10) {$q_{N}$};
			\node(n6) at (1, 0) {$p_{1}$};
			\node(n7) at (3, 0) {$p_{2}$};
			\node(n8) at (6, 0) {$p_{k}$};
			\node(n9) at (10, 0) {$p_{N}$};
			\node (n10) at (10,3)  {};
			\node (n11) at (10,4)  {};
			\node (n12) at (10,7) {};
			\node(n13) at (10,10) {};
			\node(n14) at (1, 10) {};
			\node(n15) at (3, 10) {};
			\node(n16) at (6, 10) {};
			\node(n17) at (10, 10) {};
			%  \node(n18) at (6,7){$ \begin{pmatrix}
			%  p(k) \\ q(\ell)
			%  \end{pmatrix} \in \I(S) $};
			\node(n18) at (6,7){ $\ \blacksquare $};
			\node(n21) at (7.2,8){ $ \begin{pmatrix}
				p_{k} \\ q_{\ell}
				\end{pmatrix} \in S  $};
			\node(n19) at (3,7){$\ \blacklozenge $};
			\node(n22) at (4.2,8){$ \begin{pmatrix}
				p_{2} \\ q_{\ell}
				\end{pmatrix} \in S  $};
			\node(n20) at (6,3){$ \blacklozenge $};		
			\node(n23) at (7.2,2){$ \begin{pmatrix}
				p_{k} \\ q_{1}
				\end{pmatrix} \in S  $};

			\foreach \from/\to in {n1/n2, n2/n3, n3/n4, n4/n5, 
				n1/n6, n6/n7, n7/n8, n8/n9,
				n4/n19, n19/n18, n5/n17,
				n8/n20, n20/n18, n9/n17}
			\draw[line width = 1.0] (\from) -- (\to);  
			
			\foreach \from/\to  in{n18/n16, n18/n12, n6/n14, n7/n19, n19/n15,
				n2/n20, n20/n10, n3/n11}
			\draw[dashed, line width = 0.5] (\from) -- (\to);
			
			%  \draw (0,0) -- (4,4);
			\end{tikzpicture} 
		}
	\end{subfigure}
	\caption{The figure (a) depicts the grid 
		%$ G_{S} = \{p(0) = 0, p(1), \ldots, p(N)\}\times \{q(0) = 0, q(1), \ldots, q(N)\} $ 
		$ G_{S} = \big\{0, p_{1}, \ldots, p_{N} \big\} \times 
		\big\{0, q_{1}, \ldots, q_{N}\big\} =  \{\x\cdot \eversor_{1}: \x \in S\} \times  \{\x\cdot \eversor_{2}: \x \in S\} $,
		of a particular set $ S $ satisfying Hypothesis \ref{Hyp 2D Configuration}. The corner $ p_{k}\eversor_{1} + q_{\ell}\eversor_{2} $ defines a potential ancestor, however it may or may not belong to $ S $.
		The figure (b) displays two possible ways to generate the potential ancestor of figure (a). First, when the point $ p_{k}\eversor_{1} + q_{\ell}\eversor_{2} $, marked with $ \blacksquare $ belongs to $ S $. A second option occurs when the couple of points $ p_{2}\eversor_{1} + q_{\ell}\eversor_{2} $, $p_{k}\eversor_{1} + q_{1}\eversor_{2} $, marked with $ \blacklozenge $ belong to $ S $. It is direct to see that there are $ (k -1)\times(\ell -1) + 1 $ possibilities to generate the potential ancestor at hand, but at most one of them is present in a given configuration/set.
	}
	\label{Fig Potential Ancestors}
\end{figure}
%
%
%\color{black}
%\end{comment}
%
\begin{remark}\label{Rem Non Repeated Coordinates}
	%
	%	\begin{enumerate}[(i)]
	%		\item 
	Observe that Hypothesis \ref{Hyp Non Repeated Coordinates} is mild because, it will be satisfied with probability one for any sample of $ N + 1 $ elements from $ \R^{d} $.
	%		\color{red}
	%		\item  Notice that, for any data set $ S \subseteq \R^{d} $ the partitions $ \big\{\P^{(j)}: j \in [d] \big\} $ are a monotone partition the interval $ \big(\min \pi_{j}(S), \max \pi_{j}(S) \big) $ for all $ j \in [d] $. 
	%		\color{black}
	%\item Notice that if $ \pt \in I_{n}^{(j)} $ then $ \#S_{\lt} = n $ and $ \#S_{\rt} = N + 1 - n $.
	%	\end{enumerate}
	%
\end{remark}
Next we prove that given a data set, its isolation random forest is a probability space.
\begin{theorem}\label{Thm IRF is a probability space}
	Let $ S \defining \{\x_{0}, \x_{1}, \ldots, \x_{N} \} \subseteq \R^{d}$ and let $ \irf $ be as in \textsc{Definition} \ref{Def Isolation Random Tree}. Then, the algorithm induces a probability measure in $ \irf $.
\end{theorem}
\begin{proof}
	We prove this theorem by induction on the cardinal of the set $ \#S $. For $ \#S  = N = 1 $ the only possible tree is the trivial one. 
	%For $ \#S  = N = 2 $ the result also holds. Due to Theorem \ref{Thm Well posedness of iTree} only one instance is needed to form the only possible isolation tree. 
	%(see \textsc{Remark} \ref{Rem Isolation Trees Two Data} and \textsc{Figure} \ref{Fig Isolation Tree 2 Data}). Given that the isolation tree is unique, it has probability one.
	
	Assume now that the result is true for any data set satisfying Hypothesis \ref{Hyp Non Repeated Coordinates}, with cardinal less or equal than $ N $. Let $ S = \{ \x_{0}, \ldots, \x_{N} \} $ be a set and let $ T \in \irf $ be arbitrary, such that $ j \in [d] $ was the first direction of separation, with corresponding split value $ \pt\in (\min \pi_{j}(S), \max \pi_{j}(S)) $, $ T_{\lt}, T_{\rt} $ left, right subtrees and $ S_{\lt} , S_{\rt}$ left and right sets (as in Definition \ref{Def Isolation Random Tree}). Suppose that $ \pt $ belongs to the interval $ I^{(j)}_{n} $ then, the probability that $ T $ occurs, equals the probability of choosing the direction $ j \in [d] $, times the probability of choosing $ I^{(j)}_{n} $ among $ \P^{(j)} $, times the probability that $ T_{\alpha} $ occurs in $ S_{\alpha} $ when $ \alpha\in \{\lt, \rt\} $, i.e.
	\begin{equation}\label{Eq IRT Recursive Probability}
	\prob(T) = \frac{1}{d}
	\frac{\vert I_{n}^{(j)} \vert}{\sum\{ \vert I \vert: I \in \P^{(j)} \}} 
	\prob_{\lt}(T_{\lt} )
	\prob_{\rt}(T_{\rt}) .
	\end{equation}
	Here $ \prob_{\!\alpha}(T_{\alpha} ) $ indicates the probability that $ T_{\alpha} $ occurs in the space of isolation random trees defined on the sets $ S_{\alpha} $, for $ \alpha \in \{\lt, \rt\} $; which is well-defined since $ \#S_{\alpha} \leq N $. Denoting by $ \irf(\P_{\alpha}) $ the space of isolation random trees defined on the set $ S_{\alpha} $, by the induction hypothesis we know that $ \prob_{\alpha}: \irf(S_{\alpha}) \rightarrow [0,1] $ is a well-defined probability, then $ \prob_{\!\alpha}(T_{\alpha} ) $ for $ \alpha \in \{\lt, \rt\} $ are nonnegative and consequently $ \prob(T) $ is nonnegative. Next we show that $ \sum\{\prob(T) : T \in \irf \}= 1 $. Consider the following identities
	\begin{equation}\label{Eq first equality Propability Space IRF}
	\begin{split}
	\sum_{T \in \, \irf} \prob(T) 
	& = \sum_{j \, = \, 1}^{d} 
	\sum_{\substack{T \in \, \irf\\ \pt \in \,( \min \pi_{j}(S), \max \pi_{j}(S) ) } } \prob(T) 
	%	\\
	%	& 
	= \sum_{j \, = \, 1}^{d}
	\sum_{n \, = \, 1}^{N}
	\sum_{\substack{T \in \, \irf\\ \pt \in \,  I_{n}^{(j)} } } \prob(T) \\
	& = \sum_{j \, = \, 1}^{d}
	\sum_{n \, = \, 1}^{N}
	\sum_{\substack{T \in \, \irf\\ \pt \in \,  I_{n}^{(j)} } }
	\frac{1}{d} \frac{\vert I_{n}^{(j)} \vert}{\sum\{ \vert I \vert: I \in \P^{(j)} \}} 
	\prob_{\lt}(T_{\lt} )
	\prob_{\rt}(T_{\rt}) \\
	& = \sum_{j \, = \, 1}^{d}
	\sum_{n \, = \, 1}^{N}
	\frac{1}{d} \frac{\vert I_{n}^{(j)} \vert}{\sum\{ \vert I \vert: I \in \P^{(j)} \}} 
	\sum_{\substack{T \in \, \irf\\ \pt \in \,  I_{n}^{(j)} } } 
	\prob_{\lt}(T_{\lt} )
	\prob_{\rt}(T_{\rt})
	\end{split}
	\end{equation}
	The sum nested in the third level can be written in the following way
	\begin{equation*}
	\begin{split}
	\sum_{\substack{T \in \, \irf\\ \pt \in I_{n}^{(j)}}} 
	\prob_{\lt}(T_{\lt} )
	\prob_{\rt}(T_{\rt})
	& = \sum_{\substack{T_{\lt} \in \, \irf(S_{\lt})\\T_{\rt} \in \, \irf(S_{\rt})}} 
	\prob_{\lt}(T_{\lt} )
	\prob_{\rt}(T_{\rt}) \\
	& = 
	\sum_{T_{\lt} \in \, \irf(S_{\lt})} 
	\sum_{T_{\rt} \in \, \irf(S_{\rt})}
	\prob_{\lt}(T_{\lt} )
	\prob_{\rt}(T_{\rt}) \\
	& = \sum_{T_{\lt} \in \, \irf(S_{\lt})} 
	\prob_{\lt}(T_{\lt} )
	\sum_{T_{\rt} \in \, \irf(S_{\rt})}
	\prob_{\rt}(T_{\rt}) .
	\end{split}
	\end{equation*}
	Due to the induction hypothesis, each factor in the last term equals to one. Replacing this fact in the expression \eqref{Eq first equality Propability Space IRF} we get 
	\begin{equation*}
	\begin{split}
	\sum_{T \in \, \irf} \prob(T) 
	& 
	=
	\sum_{j \, = \, 1}^{d}
	\sum_{n \, = \, 1}^{N}
	\frac{1}{d} \frac{\vert I_{n}^{(j)} \vert}{\sum\{ \vert I \vert: I \in \P^{(j)} \}} 
	%	\sum_{\substack{T \in \, \irf\\ \pt \in I_{n}^{(j)}}} 
	%	\prob_{\lt}(T_{\lt} )
	%	\prob_{\rt}(T_{\rt}) 
	%	\\
	%	& 
	= \sum_{j \, = \, 1}^{d}
	\frac{1}{d} \frac{1}{\sum\{ \vert I \vert: I \in \P^{(j)} \}}
	\sum_{n \, = \, 1}^{N} \vert I_{n}^{(j)} \vert
	%	\\
	%	%
	%	& 
	= \sum_{j \, = \, 1}^{d} \frac{1}{d} = 1,
	\end{split}
	\end{equation*}
	which completes the proof.
	%\qed
\end{proof}
\begin{corollary}\label{Thm the IRF Method Converges}
	Let $ S \defining \{\x_{0}, \x_{1}, \ldots, \x_{N} \} \subseteq \R^{d}$ and let $ \irf $ be as in \textsc{Definition} \ref{Def Isolation Random Tree}. Given a sequence of random iTree algorithm experiments (or Bernoulli trials), denoted by $ (T_{n})_{n \in \N} $ and let $ \big(\uph_{\x}(T_{n})\big)_{n\in \N} $ be the sequence of the corresponding depths for the point $ \x \in S $. Then,
	\begin{equation}\label{Stmt the IRF Method Converges}
	\frac{\uph_{\x}(T_{1}) + \uph_{\x}(T_{2}) + \ldots + \uph_{\x}(T_{n}) }{n}\xrightarrow[n\rightarrow \infty]{} \Exp(\uph_{\x}) .
	\end{equation}
	In particular, the IRF method converges and it is well-defined.
\end{corollary}
\begin{proof}
	It is a direct consequence of the Law of the Large Numbers, see \cite{BillingsleyProb}.
\end{proof}
Next we present the cardinal of the space $ \irf $.
\begin{theorem}[Cardinal of the Isolation Random Forest]\label{Thm Cardinal of RIF}
	Let $ S \defining \{\x_{0}, \x_{1}, \ldots, \x_{N} \} \subseteq \R^{d}$ and $ \irf $ be as in \textsc{Hypothesis} \ref{Def Isolation Random Tree}, then
	\begin{align}\label{Eq Cardinal of RIF}
	& \# \irf (S) \equiv \frac{1}{N}{2(N -1) \choose N -1} d^{N -1} = 
	C_{N-1} d^{N - 1}, &
	& \forall N \geq 1 .
	\end{align}
	Here $ C_{N - 1} $ denotes the $ N -1 $ Catalan number. 
\end{theorem}
\begin{proof}
	Let $ t_{i} $ be the number of all possible isolation trees on $ i $ data, with the artificial convention $ t_{0} = 0 $. It is direct to see that $ t_{1} = 1 $. Then, repeating the reasoning used to derive the expression \eqref{Eq first equality Propability Space IRF} the following recursion follows
	\begin{equation*} %\label{Eq Dimensional Catalan Recursive 1}
	\begin{split}
	%\#\irf([N + 1]) = 
	t_{N+ 1}
	= 
	\sum_{T\, \in\, \irf(S)} 1
	%&= \sum_{j\, = \, 1}^{d} \sum_{ \substack{T \, \in\, \irf\\ \pt \,\in (\min S_{j} , \max S_{j})} } 1\\
	&= \sum_{j\, = \, 1}^{d} \sum_{n\, = \, 1}^{N} \sum_{ \substack{T \, \in\, \irf\\ \pt \in \, I_{n}^{(j)} } } 1 .
	\end{split}
	\end{equation*}
	Notice that if $ \pt \in I_{n}^{(j)} $ then $ \#S_{\lt} = n $ and $ \#S_{\rt} = N + 1 - n $. Therefore, the sum $ \sum\{ 1 : T \, \in\, \irf , \,\pt \in \, I_{n}^{(j)} \} $ counts all the possible trees on $ S_{\lt} $, times the number of trees on $ S_{\rt} $; whose cardinals are $ t_{n} $ and $ t_{N + 1 - n} $ respectively. Replacing the latter in 
	%\eqref{Eq Dimensional Catalan Recursive 1} 
	the expression above, we have
	\begin{align}\label{Eq Dimensional Catalan Recursive 1}
	%\begin{split}
	%\#\irf([N + 1]) = 
	& t_{N + 1}
	= \sum_{j\, = \, 1}^{d} \sum_{n\, = \, 1}^{N} t_{n}t_{N + 1 - n} = 
	d\sum_{n\, = \, 1}^{N} t_{n}t_{N + 1 - n}
	= 
	d\sum_{n\, = \, 0}^{N + 1} t_{n}t_{N + 1 - n} , &
	& \forall N \in \N .
	%\end{split}
	\end{align}
	Let $ g(x) \defining \sum\limits_{i \, \geq \,1 } t_{i}x^{i} $ be the generating function of the sequence $ (t_{i})_{i \,\geq \, 1} $, then the relation $ d g^{2}(x) + x = g(x) $ holds which, solving for $ g(x) $ and recalling that $ g(0) = t_{0} = 0 $ gives
	\begin{equation*} %\label{Eq Dimensional Catalan Recursive 1}
	g(x) = \frac{1 - \sqrt{1 - 4dx}}{2d} .
	\end{equation*}
	The generalized binomial theorem states
	\begin{equation*} %\label{Eq Dimensional Catalan Recursive 1}
	g(x) = \frac{1}{2d}\Big(1- \sum_{k \, \geq \, 0} {1/2 \choose k}(-4dx)^{k} \Big) =
	\frac{1}{d} \sum_{k \, \geq \, 1} {1/2 \choose k}(-1)^{k + 1}2^{2k - 1}d^{k}x^{k} .
	\end{equation*}
	Recalling that 
	\begin{equation*}
	{1/2 \choose k} = \frac{(-1)^{k -1}}{2^{k}}\frac{1\cdot 3\cdot \ldots \cdot (2k-3)}{1\cdot 2 \cdot \ldots \cdot k},
	\end{equation*}
	we conclude that
	\begin{equation*}
	t_{k} = \frac{1}{k}{2(k -1)\choose k-1}d^{k -1} .
	\end{equation*}
	The above concludes the proof.
	%\qed
\end{proof}
%
%
%
%
%%%%%%%%%%%%%%%%%%%%%%%%%%%%%%%%%%%%%%%%%%%%%%%%%%%%%%%%%%%%%%%%%%%%%%%%%%%%%%%%%%
\subsection{The Inconclusiveness of the Expected Height.}\label{Sec Height Inconclusiveness}
%%%%%%%%%%%%%%%%%%%%%%%%%%%%%%%%%%%%%%%%%%%%%%%%%%%%%%%%%%%%%%%%%%%%%%%%%%%%%%%%%%
%
%
In the present section, it will be seen that the expectation of the depth, depending on the configuration of the points, can have different topological meanings when working in multiple dimensions. This is illustrated with two particular examples in 2D. Before presenting them some context needs to be introduced
\begin{hypothesis}[Adopted in the section \ref{Sec Height Inconclusiveness} only]\label{Hyp 2D Configuration}
	The data set $ S \subseteq \R^{2} $ satisfies
	\begin{enumerate}[(i)]
		\item All the data are contained in the first quadrant of the plane.
		
		\item The set $ S $ contains the origin $ \zero $.
		
		\item The set $ S $ verifies the hypothesis \ref{Hyp Non Repeated Coordinates} of \textsc{Section} \ref{Sec General Case}.
		
	\end{enumerate}
\end{hypothesis}
From now on we concentrate on analyzing the depth of the origin $ \zero $ in $ \irf $. Appealing to the regular terminology of rooted trees, a node $ u $ is the ancestor of another node $ v $, if $ u $ lies in the unique path joining the root of the tree and $ v $. In the case of isolation trees the nodes are subsets of $ S $, ergo given an isolation tree $ T $, the subsets of $ S $ in the path joining $ \zero $ with the root are all ancestors within $ T $. We say a subset $ A $ of $ S $ is a \textbf{potential ancesto}r of $ \zero $ if there exists an iTree for which $ A $ is ancestor of the origin.   
Notice that due to Hypothesis \ref{Hyp 2D Configuration}, the potential ancestors of $ \zero $ have the structure $ A = S \cap R $, where $ R \subseteq \R^{2}$ is a rectangle whose edges are parallel to the coordinate axes, see \textsc{Figure} \ref{Fig Potential Ancestors} (a). Given that infinitely many rectangles satisfy this conditions we consider $ R_{A}\defining \bigcap \{R: A = R \cap S \text{ and } R \text{ is a rectangle} \} $. Now, $ R_{A} $ can be identified with its upper right corner, moreover, given a set $ S $ with associated grid $ G_{S} = \big\{0, p_{1}, \ldots, p_{N} \big\} \times 
\big\{0, q_{1}, \ldots, q_{N}\big\}  $, we denote a potential ancestor by $ [p_{i}, q_{j}]  \defining \{\x\in S: \x \cdot \eversor_{1} \leq p_{i}, \x\cdot \eversor_{2} \leq q_{j} \} $, see \textsc{Figure} \ref{Fig Potential Ancestors}. Notice that, depending on the configuration of $ S $, not every element of $ G_{S} $ defines a potential ancestor, also observe that different configurations/sets may have an ancestor identified by the same pair, as it is the case of $ [p_{k}, q_{\ell}] $ in \textsc{Figure} \ref{Fig Potential Ancestors} (a) and (b). Finally, we introduce the following indicator function
\begin{equation*}
\X_{[p_{i}, q_{j}]} \defining \begin{cases}
1 , & [p_{i}, q_{j}] \text{ is ancestor of } \zero, \\
0 , & \text{otherwise} .
\end{cases}
\end{equation*}
It is direct to see that the depth of the origin satisfies $ \uph_{\zero} = \sum \big\{\X_{[p_{i}, q_{j}]}: [p_{i}, q_{j}] \text{ is a potential ancestor of } \zero \big\}$, as the sum counts the number of ancestors of $ \zero $ in a bijective fashion.
\begin{lemma}[Inconclusiveness of IRF in 2D]\label{Thm Inconclusiveness in 2D}
	Given a set $ S = \{\x_{0}, \x_{1}, \ldots, \x_{N}\} $ in $ \subseteq \R^{2} $, the topological meaning of the expected height, 
	%	($ \{\Exp(\uph_{\x}): \x \in S \}$)
	found by the IRF method is inconclusive and no general quality certificate can be established for the method.
\end{lemma}
\begin{proof} Let $ p_{0} = 0 < p_{1} < \ldots, p_{N} $ and $ q_{0} = 0 < q_{1} < \ldots < q_{N} $ and consider the following two sets.
	\begin{enumerate}[{Set} 1.]
		\item\label{Exm Identity Permutation} (A monotone configuration.)
		%\begin{example}[A monotone configuration]\label{Exm Identity Permutation}
		Let $ S_{1} = \{ \x_{0}, \x_{1}, \ldots, \x_{N}\} \subseteq \R^{2} $ satisfy \textsc{Hypothesis} \ref{Hyp 2D Configuration}. Let 
		%$ G_{S} = \big\{0, p(1), \ldots, p(N) \big\} \times 
		%\big\{0, q(1), \ldots, q(N)\big\}  $ 
		$ G_{S_{1}} = \big\{0, p_{1}, \ldots, p_{N} \big\} \times 
		\big\{0, q_{1}, \ldots, q_{N}\big\}  $ be its associated grid, understanding that $ p_{0} = q_{0} = 0 $ and suppose that $ \x_{i} = p_{i}\eversor_{1} + q_{i}\eversor_{2}  $ for $ i = 0, 1, \ldots, N $. In this particular case, the ancestors are identified with the points $ \x_{i} \in S $, moreover they are the upper right corners of the associated rectangles. Hence,
		\begin{equation*}
		\uph_{\zero} = \sum_{i \, \in \, [N]} 
		%\X_{[p(i), q(i)]} 
		\X_{[p_{i}, q_{i}]} .
		\end{equation*}
		Next, observe that
		\begin{equation*}
		%\Exp\big(\X_{[p(i),q(i)]}\big) = \begin{cases}
		%\tfrac{1}{2}\tfrac{p(i) - p(i - 1)}{p(i)} + 
		%\tfrac{1}{2}\tfrac{q(i) - q(i  -1)}{q(i)}, & 1 \leq i \leq N  - 1 ,\\
		%1 , & i = N .
		%\end{cases}
		\Exp\big(\X_{[p_{i},q_{i}]}\big) = \begin{cases}
		\tfrac{1}{2}\tfrac{p_{i + 1} - p_{i}}{p_{i + 1}} + 
		\tfrac{1}{2}\tfrac{q_{i + 1} - q_{i}}{q_{i + 1}}, & 1 \leq i \leq N  - 1 ,\\
		1 , & i = N .
		\end{cases}
		\end{equation*}
		%
		%Then, recalling \textsc{Theorem} \ref{Thm Expectated Heights 1D Data} we have that
		%%
		%\begin{equation}\label{Eq Identity Permutation Expectation}
		%%\Exp(\uph_0) = \O\Big( \log \Big(\frac{p(N)}{p(1)} \Big)\Big) 
		%%+  \O\Big( \log \Big(\frac{q(N)}{q(1)} \Big)\Big) = 
		%%\O\Big( \max \Big\{\log \Big(\frac{p(N)}{p(1)} \Big), \log \Big(\frac{q(N)}{q(1)} \Big) \Big\}\Big).
		%\Exp(\uph_0) = \O\Big( \log \Big(\frac{p_{N}}{p_{1}} \Big)\Big) 
		%+  \O\Big( \log \Big(\frac{q_{N}}{q_{1}} \Big)\Big) = 
		%\O\Big( \max \Big\{\log \Big(\frac{p_{N}}{p_{1}} \Big), \log \Big(\frac{q_{N}}{q_{1}} \Big) \Big\}\Big).
		%\end{equation}
		%%
		Thus, the expectation is given by 
		\begin{subequations}\label{Eq Expected Height and Distance Exm 1}
			\begin{equation}\label{Eq Expected Height Exm 1}
			%
			%\begin{split}
			\Exp\big(\uph_{\zero}\big)  
			= 1 + \sum_{i  \, = \, 1}^{N - 1} 
			\frac{1}{2}\frac{q_{i + 1} - q_{i}}{q_{i + 1}} + 
			\frac{1}{2}\frac{p_{i + 1} - p_{i}}{p_{i + 1}},
			%\end{split}
			%
			\end{equation}
			and the distance from the origin to the rest of the set is given by
			\begin{equation}\label{Eq Distance Exm 1}
			%
			%\begin{split}
			\rho_{1} \defining \dist\big(\x_{0}, S_{1} - \{\x_{0}\} \big)= 
			\sqrt{p_{1}^{2} + q_{1}^{2}} .
			%\end{split}
			%
			\end{equation}
		\end{subequations}
		%
		%\end{example}
		%
		%
		\item\label{Exm Srategic Transposition} (A strategic transposition.)
		%\begin{comment}
		%\begin{example}[A strategic transposition]\label{Exm Srategic Transposition}
		Let $ S_{2} = \{ \y_{0}, \y_{1}, \ldots, \y_{N}\} \subseteq \R^{2} $ satisfy \textsc{Hypothesis} \ref{Hyp 2D Configuration}. Let 
		%$ G_{S} = \big\{0, p(1), \ldots, p(N) \big\} \times 
		%\big\{0, q(1), \ldots, q(N)\big\}  $ 
		$ G_{S_{2}} = \big\{0, p_{1}, \ldots, p_{N} \big\} \times 
		\big\{0, q_{1}, \ldots, q_{N}\big\}  $ be its associated grid, understanding that $ p_{0} = q_{0} = 0 $ and suppose that 
		\begin{align*}
		& \y_{0} = \begin{pmatrix}
		0 \\ 0
		\end{pmatrix}, &
		& \y_{1} = \begin{pmatrix}
		p_{1} \\ q_{2}
		\end{pmatrix}, &
		& \y_{2} = \begin{pmatrix}
		p_{2} \\ q_{1}
		\end{pmatrix}, &
		& \y_{i} = \begin{pmatrix}
		p_{i} \\ q_{i}
		\end{pmatrix}, \text{ for all } i = 3, \ldots,  N .
		\end{align*}
		In this particular case, all the points $ \y_{i} \in S $ define each one a potential ancestor, but there is an additional one, the potential ancestor $ [p_{2}, q_{2}] $. Then
		\begin{equation*}
		\begin{split}
		\uph_{\zero}  & = 
		\X_{[p_{2}, q_{2}]} +
		\X_{[p_{1}, q_{2}]} + 
		\X_{[p_{2}, q_{1}]} +
		\sum_{i  \, = \, 3}^{N}\X_{[p_{i}, q_{i}]} \\
		& = \X_{[p_{1}, q_{2}]} + 
		\X_{[p_{2}, q_{1}]} +
		\sum_{i  \, = \, 2}^{N}\X_{[p_{i}, q_{i}]} .
		\end{split}
		\end{equation*}
		Computing the expectations of each function we get
		\begin{equation*}
		\Exp\big(\X_{[p_{i},q_{j}]}\big) = \begin{cases}
		\tfrac{1}{2}\tfrac{p_{2} - p_{1}}{p_{2}}, & i = 2, j = 1 ,\\
		\tfrac{1}{2}\tfrac{q_{2} - q_{1}}{q_{2}}, & i = 1, j = 2 ,\\
		\tfrac{1}{2}\tfrac{p_{i} - p_{i - 1}}{p_{i}} + 
		\tfrac{1}{2}\tfrac{q_{i} - q_{i -1}}{q_{i}}, & 2 \leq i = j \leq N  - 1 ,\\
		1 , & i = j = N .
		\end{cases}
		\end{equation*}
		Therefore,
		\begin{equation*} %\label{Eq Expected Height Exm 2}
		%
		%		\begin{split}
		\Exp\big(\uph_{\zero}\big)  
		%		& 
		= 
		\frac{1}{2}\frac{q_{2} - q_{1}}{q_{2}} + 
		\frac{1}{2}\frac{p_{2} - p_{1}}{p_{2}} +
		\sum_{i  \, = \, 2}^{N - 1} 
		\frac{1}{2}\frac{q_{i} - q_{i -1}}{q_{i}} + 
		1 +
		\frac{1}{2}\frac{p_{i} - p_{i -1}}{p_{i}} .
		%		\\
		%		& = 1 + \sum_{i  \, = \, 1}^{N - 1} 
		%		\frac{1}{2}\frac{q_{i} - q_{i -1}}{q_{i}} + 
		%		\frac{1}{2}\frac{p_{i} - p_{i -1}}{p_{i}}.
		%		\end{split}
		%
		\end{equation*}
		Hence, the expected height is given by
		\begin{subequations}\label{Eq Expected Height and Distance Exm 2}
			\begin{equation}\label{Eq Expected Height Exm 2}
			\Exp\big(\uph_{\zero}\big) = 
			1 + \sum_{i  \, = \, 1}^{N - 1} 
			\frac{1}{2}\frac{q_{i} - q_{i -1}}{q_{i}} + 
			\frac{1}{2}\frac{p_{i} - p_{i -1}}{p_{i}}
			\end{equation}
			and the distance from the origin to the rest of the set is given by
			\begin{equation}\label{Eq Distance Exm 2}
			%
			%\begin{split}
			\rho_{2}\defining
			\dist\big(\x_{0}, S_{2} - \{\x_{0}\} \big)= 
			\min\Big\{
			\sqrt{p_{1}^{2} + q_{2}^{2}},
			\sqrt{p_{2}^{2} + q_{1}^{2}} \Big\}.
			%\end{split}
			%
			\end{equation}
		\end{subequations}
		%
		%\end{example}
		%
	\end{enumerate}
	Notice that for both sets $ S_{1} $ and $ S_{2} $
	%\ref{Exm Identity Permutation} and \ref{Exm Srategic Transposition} 
	the expected height has identical value, as \textsc{Equations} \eqref{Eq Expected Height Exm 1} and \eqref{Eq Expected Height Exm 2} show. However, the topological distance from $ \mathbf{0} $ to the sets $ S_{1} - \{\zero \}, S_{2} - \{ \zero \} $ is different as \textsc{Equations} \eqref{Eq Distance Exm 1} and \eqref{Eq Distance Exm 2} show. Moreover, for simplicity assume that $ p_{1} = q_{1} $, $ p_{2} = q_{2} $ and let $ p_{1}  \rightarrow 0 $. Then, the distances behave as follows 
	\begin{align}\label{Stmt Qualitatively different limits}
	& \rho_{1} \xrightarrow[p_{1} \, \rightarrow \, 0]{} 0 , & 
	& \rho_{2} \xrightarrow[p_{1} \, \rightarrow \, 0]{} p_{2} .
	\end{align}
	Since $ p_{2} $ can take any value in $ \R $, the difference between distances can be arbitrarily large while their expected heights remain equal. In other words, in the first case the point is close to the set while in the second one $ p_{2} \in \R $ can be chosen so that $ \mathbf{0} $ becomes an anomaly.
	
	From the discussion above, it follows that although the IRF method is well-defined and it converges to $ \Exp(\uph_{\x}) $ for every $ \x \in S $ (see \textsc{Corollary} \ref{Thm the IRF Method Converges}), the topological-metric meaning of such expected value may change according to the configuration of the data. More specifically, the value $ \Exp(\uph_{\x}) $ is inconclusive from the topological-metric point of view and therefore, its reliability to asses whether or not a point is an anomaly, is uncertain. Moreover, the analysis of limits in the expression \eqref{Stmt Qualitatively different limits} discussed above, shows that no general quality certificate about the method can be given.
\end{proof}
\begin{remark}\label{Rem Inconclusiveness of the IRF Method}
	A third example can be constructed similarly to the sets $ S_{1}, S_{2} $
	%\ref{Exm Identity Permutation} and \ref{Exm Srategic Transposition} 
	in the proof of \textsc{Lemma} \ref{Thm Inconclusiveness in 2D}. Let $ S = \{ \z_{0}, \z_{1}, \ldots, \z_{N}\} \subseteq \R^{2} $ be given by 
	\begin{align*}
	& \z_{1} = \begin{pmatrix}
	p_{1} \\ q_{N}
	\end{pmatrix}, &
	& \z_{i} = \begin{pmatrix}
	p_{i} \\ q_{i  - 1}
	\end{pmatrix}, \text{ for all } i = 2, \ldots,  N .
	\end{align*}
	Then, in this configuration $ \Exp(\uph_{0}) $ satisfies the identity \eqref{Eq Expected Height Exm 2}, while the distance from the origin to the rest of the set is given by
	\begin{equation*} %\label{Eq Distance Exm Rem}
	%
	%\begin{split}
	\rho_{3}\defining
	\dist\big(\x_{0}, S - \{\x_{0}\} \big)= 
	\min\Big\{
	\sqrt{p_{1}^{2} + q_{N}^{2}},
	\sqrt{p_{2}^{2} + q_{1}^{2}} \Big\}.
	%\end{split}
	%
	\end{equation*}
	This third example ads even more inconclusiveness to the IRF method in multiple dimensions, on top of that detected by the analysis presented in \textsc{Lemma} \ref{Thm Inconclusiveness in 2D}.	
	%		\item The latter has been shown for the 2D case, but it is natural to expect similar issues for higher dimensions and significantly higher complexity. 
	
	%		\item Another aspect that could be stressed is the fact that the examples \ref{Exm Identity Permutation} and \ref{Exm Srategic Transposition} are extreme cases, measured by the number of potential ancestors (see \textsc{Corollary} \ref{Thm Expected Number of Potential Ancestors General Case}). A particular analysis could be proposed for data sets with the same number (or close) of potential ancestors. To that end, we compute the expected number of potential ancestors in \textsc{Theorem} \ref{Thm Expected Number of Potential Ancestors} and Corollary \ref{Thm Expected Number of Potential Ancestors General Case}. On the other hand, it has been shown in \textsc{Section} \ref{Sec Estimates Expected Height} that, depending on the configuration and the values of its data, even proposing bounds for $ \Exp(X[p_{k}, q_{j}]) $ is not possible. 
	
	%	\end{enumerate}
	%
\end{remark}
\begin{theorem}[Inconclusiveness of the IRF Method]\label{Thm Inconclusiveness IRF Method}
	Given a set $ S = \{\x_{0}, \x_{1}, \ldots, \x_{N}\} $ in $ \subseteq \R^{d} $, the topological meaning of the expected height, 
	%	($ \{\Exp(\uph_{\x}): \x \in S \}$)
	found by the IRF method is inconclusive and no general quality certificate can be established for the method.
\end{theorem}
\begin{proof}
	We proceed as in the proof of \textsc{Lemma} \ref{Thm Inconclusiveness in 2D}. In particular the result has been already shown for $ d = 2 $ then, from now on we assume that $ d \geq 3 $.
	
	Consider a sequence of numbers $ 0 = r_{0} < r_{1} < \ldots < r_{n} $, 
	%recall the sets $ S_{1}, S_{2} $ introduced in the proof of  \textsc{Lemma} \ref{Thm Inconclusiveness in 2D} and for $ k = 1, 2 $, 
	consider the sets $ \tilde{S}_{1} \defining \{\tilde{\x}_{0}, \tilde{\x}_{1}, \ldots, \tilde{\x}_{N}\} $, $  \tilde{S}_{2} \defining \{\tilde{\y}_{0}, \tilde{\y}_{1}, \ldots, \tilde{\y}_{N}\} $, whose elements are given by
	\begin{align*}
	%	&, &
	& \tilde{\x}_{i} \defining r_{i}\sum\limits_{j \, \in \, [N] }\eversor_{j}, &
	& \forall i \in [N], \\
	& \tilde{\y}_{1} = \frac{ r_{1} }{ 2 }  \, \eversor_{1} + \frac{ r_{1} }{ 4 } \, \eversor_{2}  
	+ r_{1} \sum\limits_{j \, = \, 3 }^{N}\eversor_{j},&
	& \tilde{\y}_{2} = \frac{ r_{1} }{4} \, \eversor_{1} + \frac{ r_{1} }{2} \, \eversor_{2}  
	+ r_{2} \sum\limits_{j \, = \, 3 }^{N}\eversor_{j} \\
	%	&, &
	& \tilde{\y}_{i} \defining r_{i}\sum\limits_{j \, \in \, [N] }\eversor_{j}, &
	& \forall i \in \{ 0, 3, 4, \ldots, N \}, .
	\end{align*}
	Consider the projection map $ \pi: \R^{N} \rightarrow \R^{2} $ defined by $ \pi(\x) \defining \pi(\x) = 
	(\x\cdot \eversor_{1}) \, \eversor_{1} + (\x\cdot \eversor_{2}) \, \eversor_{2} $. It is direct to see that $ \pi(\tilde{S}_{1}),  \pi(\tilde{S}_{2}) $, have the same structure of the sets $ S_{1}, S_{2} $ introduced in the proof of \textsc{Lemma} \ref{Thm Inconclusiveness in 2D}, hence, they have their same properties. Consequently, the element $ \zero $ in the sets $ \tilde{S}_{1}, \tilde{S}_{2} $ has the same the number and description of potential ancestors as in $ S_{1}, S_{2} $ respectively. Therefore, the values of the expected height for the point $ \zero $ will agree, but the distances from it to the nearest point in $ S_{1} - \{ \zero\}, S_{2} - \{ \zero \} $ can arbitrarily disagree as already shown in \textsc{Lemma} \ref{Thm Inconclusiveness in 2D}. 	
\end{proof}
\section{The 1D Setting}\label{Sec 1D Setting}
%%%%%%%%%%%%%%%%%%%%%%%%%%%%%%%%%%%%%%%%%%%%%%%%%%%%%%%%%%%%%%%%%%%%%%%%%%%%%%%%%%%%%%%%%%
%
%
The present section has several objectives. 
\begin{itemize}
	\item Show that the IRF method \textbf{is conclusive} in the 1D setting (\textsc{Section} \ref{Sec bound b(H)}). 
	
	\item Estimate de variance of $ \irf $ and use it to derive confidence intervals for the method in the general setting (\textsc{Section} \ref{Sec Variance MRF}). 
	
	\item Introduce an alternative version of the IRF method, based on the mathematical robustness of the 1D case (\textsc{Section} \ref{Sec Correction of the IRF Method}).  
\end{itemize}
%
%$ \bullet $ Show that the IRF method \textbf{is conclusive} in the 1D setting (\textsc{Section} \ref{Sec bound b(H)}). $ \bullet $ Estimate de variance of $ \irf $ and use it to derive confidence intervals for the method in the general setting (\textsc{Section} \ref{Sec Variance MRF}). $ \bullet $ Introduce an alternative version of the IRF method, based on the mathematical robustness of the 1D case (\textsc{Section} \ref{Sec Correction of the IRF Method}).  

In order to study the problem for sets $ S = \{ x_{0}, x_{1}, \ldots, x_{N}\} $ in $ \R $, it is strategic to start analyzing another well-known related problem analyzed in \cite{SeidelAragornRST}. We begin introducing some definitions.
\begin{definition}\label{Def Monotone Partition}
	It will be said that the family $ \P = \{I_{n}: n\in [N]  \} $ is a \textbf{monotone partition} of an interval $ [a,b] $, if there exists a monotone sequence $ a = x_{0} < x_{1} < \ldots < x_{N} = b $, such that the extremes of $ I_{n} $ are $ x_{n-1} $ and $ x_{n} $ for all $ n \in [N] $. (See \textsc{Figure} \ref{Fig Example MRF Tree} (a) for an example with $ \vert \P \vert = 5 $.)
	%%
	%\begin{enumerate}[(i)]
	%\item $ \bigcup\limits_{n\,\in\, [N]} I_{n} = [a,b] $.
	%
	%\item $ I_{n} $  is a nonempty interval for all $ n\in [N] $.
	%
	%\item $ b_{n} = a_{n + 1} $ for all $ n \in [N-1] $, where $ a_{n}, b_{n} $ denote the extremes of the interval $ I_{n} $.
	%\end{enumerate}
	%%
\end{definition}
\begin{definition}\label{Def Monotone Random Tree}
	Let $ \P = \{I_{n}: n\in [N]  \} $ be a monotone partition of the interval $ I = [a, b] $, with endpoints set $ x_{0} = a < x_{1}< \ldots < x_{N} = b $, from now denoted by $ S(\P) $.
	%, such that its lengths are all integer, i.e. $ b_{n} - a_{n} = \vert I_{n}\vert  \in \N $ for al $ n \in [N] $
	%
	\begin{enumerate}[(i)]
		\item A \textbf{monotone random tree} $ T $, associated to this partition is defined recursively as follows (see \textsc{Figures} \ref{Fig Example MRF Tree} (a) and (b) for an example.):
		\begin{enumerate}[a.]
			\item If the partition is non-empty, choose $ \troot(T) = I_{k} \in \P $ randomly as the root of $ T $, with probability $  \dfrac{\vert I_{k}\vert }{\vert I \vert } = \dfrac{x_{k - 1} - x_{k}}{b - a}$.
			
			\item Perform a monotone random tree on the \textbf{left partition} of intervals $ \P_{\lt} \defining \{I_{n}: n\in [N], n < k \} $, denoted by $ T_{\lt} $. Next, include the arc $ \big(\troot(T), \troot(T_{\lt})\big) $ in the edges of the tree $ E(T) $ and label it with $ \{x_{j}: 0 \leq j < k \} $. Here, $ \troot(T_{\lt}) $ indicates the root of $ T_{\lt} $. %Label this arc with $ S(\P_{\lt}) $.
			
			\item Perform a monotone random tree on the \textbf{right partition} of intervals $ \P_{\rt} \defining \{I_{n}: n\in [N], n > k \} $, denoted by $ T_{\rt}$. Next, include the arc $ \big(\troot(T), \troot(T_{\rt}) \big) $ in the edges of the tree $ E(T) $ and label it with $ \{x_{j}: k \leq j \leq N \} $. Here, $ \troot(T_{\rt}) $ indicates the root of $ T_{\rt} $.  %Label this arc with $ S(\P_{\rt}) $.
		\end{enumerate}
		\item We denote the set of all possible monotone random trees associated to the partition $ \P $ by $ \mrf(\P) $. Whenever the context is clear, we simply write $ \mrf $ and we refer to it as the \textbf{monotone random forest}.
	\end{enumerate}
\end{definition}
\begin{remark}\label{Rem Monotone Random Tree}
	Labeling the edges of the monotone random tree is done to ease, later on, the connection between the monotone trees and the iTrees for the 1D case.
\end{remark}
%\begin{algorithm} %[H]
%\caption{Random Tree Intervals 1D, returns a list list of $ N -1 $ arcs defining a directed rooted tree on the vertices $  \{I_{n}: n\in [N] \} $ }
%\label{Alg Random Tree Intervals 1D}
%\begin{algorithmic}[1]
%\Procedure{Random Tree Generator}{Ordered Partition $ \{ I_{n}: n\in [N] \} $ of interval $ [a, b] $
%}
%\State $ \text{ABT} = \emptyset $ \Comment{Initializing the list of Arcs in the Binary Tree as empty} 
%\State $ \troot = \emptyset $ \Comment{Initializing the root as empty} 
%\Function{Branch}{$ \{I_{n}: n\in [N]  \}, \text{ABT}, \troot $}
%\If {$ N = 1 $ } \Comment{Checking when to stop}
%\If {$ \troot  \neq  \emptyset $}
%%\State $ \troot \defining I_{1} $
%%\Else
%\State $ (\troot, I_{1}) \rightarrow \text{ABT} $ \Comment{Push the arc $ (\troot, I_{1}) $ to \text{ABT} } 
%\EndIf
%%\State \Return ($ \text{ABT}, \text{root} $)
%\State \Return $ \text{ABT}$
%\Else
%\State choose $ R \in [N] $ randomly
%\State $ (\troot, I_{\rt}) \rightarrow \text{ABT} $ \Comment{Push the arc $ (\troot, I_{\rt}) $ to \text{ABT} } 
%\State Define $ \text{root} \defining  I_{\rt} $.
%\State \textsc{Branch}$ (\{I_{n}: n\in [N], n < R \}, \text{ABT}, \troot ) $
%%\EndFunction
%\State \textsc{Branch}$ (\{I_{n}: n\in [N], n > R \}, \text{ABT}, \troot ) $
%%\EndFunction
%\EndIf
%\EndFunction
%\EndProcedure
%\end{algorithmic}
%\end{algorithm}
%
%
%
%
%
\begin{figure}
	\centering
	\begin{subfigure}[Monotone partition $ \P = \{I_{n}: 1 \leq n \leq 5 \} $.]
		{
			\begin{tikzpicture}
			[scale=.57,auto=left,every node/.style={}]
			\node (n1) at (0, 0) {$ x_{0}  = a $};
			\node (n2) at (3.5,0)  {$x_{1}  $};
			\node (n3) at (5.5,0)  {$ x_{2} $};
			\node (n4) at (8,0)  {$ x_{3} $};
			\node (n5) at (13,0)  {$ x_{4} $};
			\node (n6) at (15,0)  {$ x_{5} = b $};

			\node (n8) at (0, 1) {$ \vert $};
			\node (n9) at (3.5,1)  {$ \vert  $};
			\node (n10) at (5.5,1)  {$ \vert $};
			\node (n11) at (8,1)  {$ \vert $};
			\node (n12) at (13,1)  {$ \vert $};
			\node (n13) at (15,1)  {$ \vert $};

			\node (n15) at (1.5, 2) { $ I_{1} $};
			\node (n16) at (4,2)  {$ I_{2}  $};
			\node (n17) at (6.5,2)  {$ I_{3} $};
			\node (n18) at (10.5,2)  {$ I_{4} $};
			\node (n19) at (14,2)  {$ I_{5} $};
			
			\node (n20) at (10.4, -2.2) { \textcircled{\footnotesize 1}};
			\node (n21) at (1.4,-1.7)  {   \textcircled{\footnotesize 2}};
			\node (n22) at (7.4,-1.2)  { \textcircled{\footnotesize 3}};
			\node (n23) at (4.4,0.3)  {  \textcircled{\footnotesize 4}};
			\node (n24) at (14,-1.7)  { \textcircled{\footnotesize 2}};

			\foreach \from/\to in {n8/n9,n9/n10, n10/n11, n11/n12,
				n12/n13}
			\draw[line width = 1.0] (\from) -- (\to);

			\draw[<-, dashed, line width = 0.7] (10,-2) -- (10,3);
			\draw[<-, dashed, line width = 0.7] (1,-1.5) -- (1,2.5);
			\draw[<-, dashed, line width = 0.7] (7,-1) -- (7,2);
			\draw[<-, dashed, line width = 0.7] (4,0.5) -- (4,1.5);
			\draw[<-, dashed, line width = 0.7] (13.6,-1.5) -- (13.6,2.5);
			
			%\foreach \from/\to  in{ n11/n12}
			%\draw[dashed] (\from) -- (\to);
			
			\end{tikzpicture}
		}
	\end{subfigure}
	\begin{subfigure}[Structure $ T $, possible monotone random tree on the monotone partition $ \P $ ]
		{
			\begin{tikzpicture}
			[scale=.57,auto=left,every node/.style={}]
			\node (n4) at (5,8)  {$ I_{4} $};
			\node (n1) at (1,6)  {$ I_{1} $};
			\node (n3) at (3.0,3)  {$ I_{3} $};
			\node (n2) at (5,0)  {$ I_{2} $};
			\node (n5) at (9,6)  {$ I_{5} $};
			\node (n6) at (1,7.5)  {$ \{ x_{0}, x_{1}, x_{2}, x_{3}\} $};
			\node (n7) at (8, 7.5)  {$ \{x_{4}, x_{5}\} $};
			\node (n8) at (3.5, 4.5)  {$ \{x_{1}, x_{2}, x_{3} \} $};
			\node (n9) at (2.8, 1.5)  {$ \{x_{2}, x_{3} \} $};		
			%			\node(n18) at (6,7){ $\stackrel{ \blacksquare }{\begin{pmatrix}
			%					p_{k} \\ q_{\ell}
			%					\end{pmatrix} \in S  }$};
			%			\node(n19) at (3,7){$\stackrel{ \blacklozenge }{\begin{pmatrix}
			%					p_{2} \\ q_{\ell}
			%					\end{pmatrix} \in S  }$};
			%			\node(n20) at (6,3){$\stackrel{ \blacklozenge }{\begin{pmatrix}
			%					p_{k} \\ q_{1}
			%					\end{pmatrix} \in S  }$};		
			
			\foreach \from/\to in {n4/n1, n1/n3, n3/n2, 
				n4/n5}
			\draw[line width = 1.0] (\from) -- (\to);  
			%			
			%			\foreach \from/\to  in{n18/n16, n18/n12, n6/n14, n7/n19, n19/n15,
			%				n2/n20, n20/n10, n3/n11}
			%			\draw[dashed] (\from) -- (\to);

			\end{tikzpicture}
		}
	\end{subfigure}
	\caption{Example of a generated monotone random tree. Figure (a) displays a partition of 5 intervals. The circled numbers represent the instance when each interval was chosen, understanding that $ I_{1}, I_{5} $ have both number 2 since they belonged to left and right sub-tree after the first choice $ I_{4} $. Figure (b) depicts the monotone random tree $ T $ that is formed after the choices made in the figure (a); observe the labeling of the edges.
	}
	\label{Fig Example MRF Tree}
\end{figure}
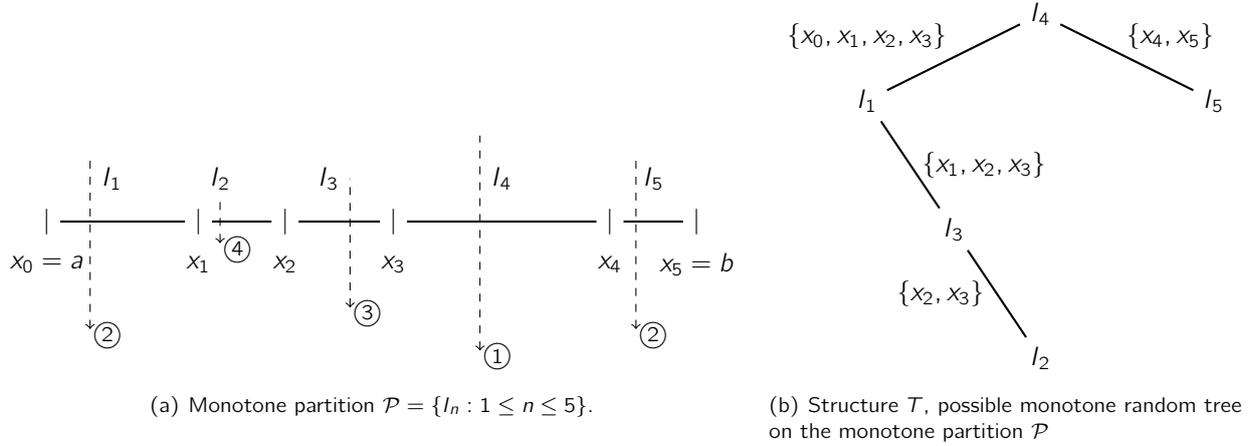
Next we recall a classic definition, see \cite{GrossYellen}
\begin{definition}\label{Def left and right subtree}
	Let $ T $ be a binary tree, the \textbf{left} (resp. \textbf{right}) \textbf{subtree} of a vertex $ v $ is the binary subtree spanning the left (resp. right)-child of $ v $ and all of its descendants. 
\end{definition}
%
%
%
%%
%\begin{proposition}\label{Thm MRT contains binary trees}
%	Let $ T $ be in $ \mrf $, then $ T $ it is a binary tree.
%	%
%	%\begin{enumerate}[(i)]
%	%\item $ T $ is a tree.
%	%
%	%\item $ T $ is binary.
%	%\end{enumerate}
%	%%
%\end{proposition}
%%
%%
%\begin{proof}
%	Analogous to the proof of 
%	%We proceed by induction on the number of elements in the partition $ \#\P = N $. For $ \# \P= N = 1 $, the result is trivial since there is only one possible tree. For $ \#\P= N = 2 $, the result also holds since there are only two possible trees. Assume now that the result holds whenever the monotone partition has cardinal less or equal than $ N -1 $. Take $ \P= \{ I_{n}: n \in [N] \} $ and $ T \in \mrf $ with $ \troot(T) $, $ \P_{\lt}, \P_{\rt} $ its root, left and right partitions respectively. Given that $ \#\P_{\lt} < N $ and $ \#\P_{\rt} < N $, then $ T_{\lt} $ and $ T_{\rt} $ are both binary trees (one of them may be empty), with corresponding roots $ \troot(T_{\lt}) $ and $ \troot(T_{\rt}) $. Since the arcs $ (\troot(T), \troot(T_{\lt})) $ and $ (\troot(T), \troot(T_{\rt})) $ are included in the list of edges $ E(T) $ by construction (whenever its respective tree is nonempty), it follows that $ T $ is a binary tree.
%	%\qed
%\end{proof}
%%
%
%Next we recall a classic definition, see \cite{GrossYellen}
%
%
\begin{theorem}\label{Thm MRT is a probability space}
	Let $ \P = \big\{I_{n}: n\in [N]  \big\} $ and $ \mrf $ be as in \textsc{Definition} \ref{Def Monotone Random Tree} then, the algorithm induces a probability measure in $ \mrf $.
\end{theorem}
\begin{proof}
	Analogous to the proof of \textsc{Theorem} \ref{Thm IRF is a probability space}.
\end{proof}
%
%
%\color{red}
%\begin{remark}\label{Rem Cardinal of Random Forest}
%Observe that the space $ \mrf (\P) $ with $ \#\P = i $, is made of all possible binary trees on $ i $ vertices; to ease notation we write $ \mrf([i]) $.  Denoting by $ b_{i} $ the number of all possible binary trees on $ i $ vertices, with the artificial convention that $ b_{0} = 1 $, then $ \#\mrf([i]) = b_{i} $. Moreover, the set of all possible binary trees satisfies the Catalan recursion
%%
%\begin{equation}\label{Eq Catalan Recursive}
%\#\mrf([N]) = b_{N} 
%=\sum\limits_{i\, = \, 1}^{N -1}b_{i} b_{N - 1 -i} 
%= \sum\limits_{i\, = \, 1}^{N -1}\#\mrf([i]) \#\mrf([N - 1 -i]) .
%\end{equation}
%% 
%Therefore, the cardinal of the space is given by 
%%
%\begin{align}\label{Eq Cardinal of Random Forest}
%& \#\mrf(\P) \equiv \frac{1}{N+1}{2N \choose N}, &
%& \#\P = N .
%\end{align}
%%
%For the proof of \eqref{Eq Catalan Recursive} and \eqref{Eq Cardinal of Random Forest}, see Section 3.8.1 in \cite{GrossYellen}. Finally, we notice that the relationship \eqref{Eq Recursive Probability} defines how likely to occur is a binary tree in the monotone random forest. Therefore, the space $ \mrf(\P) $ is not endowed with the uniform probability $ \prob_{0} (T) = \frac{1}{\#\mrf(\P) } $ for all $ T \in \mrf(\P) $, but with a function heavily dependent on the relative sizes of the intervals. Moreover, $ \prob $ is the uniform probability function, if and only if, all the intervals have the same length.
%\end{remark}
%\color{black}
%
%
\begin{definition}\label{Def Main Paramenters 1D}
	Let $ \P = \{I_{n}: n\in [N]  \} $, $ \mrf $ be as in \textsc{Definition} \ref{Def Monotone Random Tree} and let $ i \in [N] $ be fixed. Define 
	\begin{enumerate}[(i)]
		\item $ \uph_{i} : \mrf \rightarrow \N \cup\{0\} $, $ \uph_{i}(T) = $ the depth of the interval $ I_{n} $ in the tree $ T \in \mrf $.
		
		\item Given $ k \in [N] - \{i\}$, define $ \X_{i, k}: \mrf \rightarrow \{0, 1\} $ as $ \X_{i, k}(T) = 1 $ if $ I_{k} $ is ancestor of $ i $ in $ T $ and $ \X_{i, k}(T) = 0 $ otherwise.
		
		\item For $ k \neq i $ define the quantity 
		\begin{equation}\label{Def Absoltue Weights}
		w_{i, k} \defining  \sum\limits_{ \ell  \, = \, \min\{i, k \}}^{ \max\{i, k\} } \vert I_{\ell} \vert .
		\end{equation} %$ w_{i, k} \defining  \sum\limits_{ \ell  = \min\{i, k \}}^{ \max\{i, k\} }$
	\end{enumerate}
\end{definition}
Notice that while $ \uph_{i} $, $ \X_{i, k} $ are random variables depending on the tree, the quantities $ w_{i,k} $ are not, moreover $ w_{i, k} = w_{k, i} $.
\begin{lemma}\label{Thm Expectation of Random Variables}
	Let $ \P = \{I_{n}: n\in [N]  \} $ %$ I $ 
	and $ \mrf $ be as in \textsc{Definition} \ref{Def Monotone Random Tree}. Then
	\begin{enumerate}[(i)]
		\item The height of the interval $I_{i}$ satisfies $ \uph_{i} = \sum\limits_{k \neq i } \X_{i,k} $.
		
		\item $ \Exp(\X_{i, k}) = \dfrac{\vert I_{k} \vert }{w_{i,k}} $, with $ w_{i,k} $ as defined in \eqref{Def Absoltue Weights}.
		
		\item 
		\begin{equation}\label{Ineq Expected Height 1D}
		\Exp(\uph_{i}) = \O\Big(\log \Big( \dfrac{\sum\{\vert I\vert: I \in \P\}}{\vert I_{i}\vert} \Big) \Big) .
		\end{equation}
		%\item Suppose that the length of every interval is an integer i.e.,  $ \vert I_{n}\vert  \in \N $ for al $ n \in [N] $. Then, the expected depth of the interval $ I_{i} $ is $ \O\big(\log ( \dfrac{\sum\{\vert I\vert: I \in \P\}}{\vert I_{i}\vert} ) \big) $.
	\end{enumerate}
\end{lemma}
\begin{proof}
	See \cite{SeidelAragornRST}.
\end{proof}
\subsection{The Quality of the Bound $ b(\uph) $}\label{Sec bound b(H)}
%%%%%%%%%%%%%%%%%%%%%%%%%%%%%%%%%%%%%%%%%%%%
%
%
In this section we discuss the quality of the bound \eqref{Ineq Expected Height 1D}. To that end, we introduce the quantities  
\begin{align}\label{Ineq bound Expected Height 1D}
& b(\uph_{i})  \defining \log \dfrac{\sum\{\vert I \vert: I \in \P \}}{\vert I_{i} \vert}, &
& i = 1, \ldots , N ,
\end{align}
%.
for a monotone partition $ \P = \{ I_{i}: 1 \leq i \leq N \} $.
\begin{theorem}\label{Thm Bound Quality}
	Consider a monotone partition $ \P =  \{I_{n}: n \in [N] \} $ defined by the sequence of points 
	\begin{align*}
	& S \defining \big\{ 0, \xi, x_{2}, \ldots, x_{N - 1}, 1  \big\}, &
	& 0< \xi  < x_{2} <  \ldots <  x_{N - 1} < 1 .
	\end{align*}
	Then,
	\begin{align}\label{Stmt Bound Quality Limits}
	&  \frac{\Exp\big( \uph_{1}\big)}{b\big(\uph_{1}\big)} \xrightarrow[\xi \, \rightarrow \, 0]{} 0 , &
	&  \frac{\Exp\big( \uph_{1}\big)}{b\big(\uph_{1}\big)} \xrightarrow[\xi \, \rightarrow \, 1]{} 1 .
	\end{align}
	Consequently, the quality of the bound $ b(\uph_{i}) $ improves as the point becomes more of an outlier and it deteriorates as the point gets closer to the cluster of points. Hence, $ \Exp(\uph_{i}) $ contains reliable topological information for outliers, though its quality of information is poor for cluster points.
	(Notice that there are no conditions for $  x_{k} $ with $ k = 2, \ldots, N -1 $, other than the monotonicity and the boundedness detailed above.) 
\end{theorem}
\begin{proof}
	First observe that
	\begin{equation}\label{Eq Algebraic Form of Logarithmic Estimate}
	b(\uph_{1}) = \log\dfrac{1}{d} = 
	-\log d .
	\end{equation}
	Recall that the height of the interval $ I_{1} $ is given by $ \uph_{1} = \sum\limits_{k \, = \, 2}^{N} \X_{1, k} $. From \textsc{Lemma} \ref{Thm Expectation of Random Variables} we have $ \Exp(\X_{1, k}) = \dfrac{x_{k} - x_{k - 1}}{x_{k}} $ then, $ 	\Exp\big(\uph_{1}\big) = \sum\limits_{k \, = \, 2}^{N} \dfrac{x_{k} - x_{k - 1}}{x_{k}} $.
	%	%
	%	\begin{equation*}
	%	\Exp\big(\uph_{1}\big) = \sum\limits_{k \, = \, 2}^{N} \dfrac{x_{k} - x_{k - 1}}{x_{k}}
	%	.
	%	\end{equation*}
	%	%
	Now consider the following estimates
	\begin{subequations}
		\begin{equation}\label{Ineq Arithmetic Geometric Mean Upper Bound Height Expectation}
		\Exp\big(\uph_{1}\big) %= \sum\limits_{k \, = \, 2}^{N} \dfrac{x_{k} - x_{k - 1}}{x_{k}}
		= ( N - 1 )  -  \sum\limits_{k \, = \, 2}^{N} \dfrac{x_{k - 1}}{x_{k}}
		\leq (N - 1) - ( N - 1)\Big( \prod\limits_{k \, = \, 2}^{N} \dfrac{x_{k - 1}}{x_{k}}  \Big)^{\frac{1}{N- 1}}
		=( N - 1 )( 1 - \xi^{ \frac{1}{N-1} } ) .
		\end{equation}
		Here, the arithmetic-geometric mean inequality was used, together with the telescopic properties of the product left. Next, observing that $ x_{k} \geq \xi $ for all $ k = 2, \ldots, N $ and using telescopic properties of the remaining terms we get
		\begin{equation}\label{Ineq Direct Upper Bound Height Expectation}
		\Exp\big(\uph_{1}\big) 
		= \sum\limits_{k \, = \, 2}^{N} \Big( 1 - \dfrac{ x_{k - 1}}{x_{k}}\Big)
		\leq  \sum\limits_{k \, = \, 2}^{N} \dfrac{x_{k} - x_{k - 1}}{\xi}
		= \frac{1 - \xi}{\xi} .	
		\end{equation}
		Again, using that $ x_{k} \leq 1 $ for all $ k = 2, \ldots , N $ combined with the telescopic properties of the remaining terms we get
		\begin{equation}\label{Ineq Lower Bound Height Expectation}
		\Exp\big(\uph_{1}\big) = \sum\limits_{k \, = \, 2}^{N} \dfrac{x_{k} - x_{k - 1}}{x_{k}}
		\geq  \sum\limits_{k \, = \, 2}^{N}\dfrac{x_{k} - x_{k - 1}}{1} 
		=1 - \xi .
		\end{equation}
	\end{subequations}
	Combining \eqref{Eq Algebraic Form of Logarithmic Estimate} with the upper bound \eqref{Ineq Arithmetic Geometric Mean Upper Bound Height Expectation} we get
	\begin{equation*}
	0 \leq 
	%	\frac{d - 1}{\log d}\leq 
	\frac{\Exp\big( \uph_{1}\big)}{b\big(\uph_{1}\big)} 
	\leq 
	\frac{(N - 1) (1 - \xi^{ \frac{1}{N-1} } ) }{ - \log \xi} .
	\end{equation*}
	Letting $ \xi \rightarrow 0 $, the upper bound above delivers the first limit in \eqref{Stmt Bound Quality Limits}. Combining \eqref{Eq Algebraic Form of Logarithmic Estimate} with the upper and lower bounds \eqref{Ineq Direct Upper Bound Height Expectation}, \eqref{Ineq Lower Bound Height Expectation} we get
	\begin{equation*}
	\frac{\xi - 1}{\log \xi}\leq \frac{\Exp\big( \uph_{1}\big)}{b\big(\uph_{1}\big)} 
	\leq 
	\frac{\xi - 1}{\xi \, \log \xi} .
	\end{equation*}
	Letting $ \xi \rightarrow 1 $, the second limit in \eqref{Stmt Bound Quality Limits} follows.
\end{proof}
\subsection{Variance and Confidence Intervals of the Monotone Random Forest $ \mrf $}\label{Sec Variance MRF}
%%%%%%%%%%%%%%%%%%%%%%%%%%%%%%%%%%%%%%%%%%%%
%
%
%
%
In the present section we estimate the variance of the heights through the monotone random forest and use this information to give a number of Bernoulli trials (random sampling) in order to guarantee a confidence interval, endowed with a confidence level, for the computed value of the expected height.
%
%\color{blue}
%
\begin{theorem}\label{Thm Variance Heights 1D}
	Let $ \P = \{I_{n}: n\in [N]  \} $ %$ I $ 
	and $ \mrf $ be as in \textsc{Definition} \ref{Def Monotone Random Tree}. Then
	\begin{subequations}\label{Eq Variance Heights 1D Equality and Estimate}
		\begin{equation}\label{Eq Variance Heights 1D}
		%& 
		\Var(\uph_{i}) = 
		\sum\limits_{k \neq i} \Exp(\X_{i, k}) +
		\sum\limits_{k \neq i}\Exp(\X_{i, k}) \sum\limits_{\substack{\ell \neq i\\\ell \neq k}} \Exp( \X_{i, \ell})
		- \Exp^{2}(\uph_{i}) ,
		% &
		%& \text{ for all } i = 0, \ldots, N 
		\end{equation}
		\begin{align}\label{Ineq Variance Heights 1D}
		%& 
		\Var(\uph_{i}) \leq \Exp(\uph_{i}), 
		%&
		%& \text{ for all } i = 0, \ldots, N .
		\end{align}
		for all $ i = 0, \ldots, N $.
	\end{subequations}
\end{theorem}
\begin{proof}
	Recall that $ \Var(\uph_{i} ) =  \Exp(\uph_{i}^{2}) - \Exp^{2}(\uph_{i}) $ and that $ \uph_{i} = \sum\limits_{k \neq i} \X_{i, k} $, then 
	\begin{equation}\label{Eq Variance Heights 1D Basic Equality}
	\begin{split}
	\Var(\uph_{i}) & = \Exp(\uph_{i}^{2})  - \Exp^{2}(\uph_{i}) =
	\sum\limits_{k \neq i} \Exp(\X_{i, k}) +
	\sum\limits_{k \neq i} \sum\limits_{\substack{\ell \neq i\\\ell \neq k}} \Exp(\X_{i, k}  \X_{i, \ell})
	- \Exp^{2}(\uph_{i}) .
	\end{split}
	\end{equation}
	%
	%$ \uph_{i}^{2} = \sum\limits_{k \neq i} \sum\limits_{\ell \neq i}\X_{i, k}  X_{i, \ell} $. 
	In order to analyze the independence of the random variables involved in the expression above we proceed by cases
	\begin{subequations}
		\begin{multline*}
		k < i < \ell \text{ or } \ell < i < k: \\
		\prob(\X_{i, k} \X_{i, \ell} = 1) = \prob(\X_{\ell, k} = 1) \prob(\X_{i, \ell} = 1), 
		\end{multline*}
		\begin{multline*}
		k < \ell < i \text{ or } i < \ell < k: \\
		\Exp(\X_{i, k} \X_{i, \ell}) =
		\prob(\X_{i, k} \X_{i, \ell} = 1) = 
		\prob(\X_{i, k} = 1) \prob(\X_{i, \ell} = 1 \vert \X_{i, k} = 1) = \Exp(\X_{i,k}) \Exp(\X_{i, \ell}),
		\end{multline*}
		\begin{multline*}
		\ell < k < i \text{ or } i< k < \ell: \\
		\Exp(\X_{i, k} X_{i, \ell}) = \prob(\X_{i, k} \X_{i, \ell} = 1) = 
		\prob(\X_{i, \ell} = 1) \prob(\X_{i, k} = 1 \vert \X_{i, \ell} = 1) = \Exp(\X_{i,\ell}) \Exp(\X_{i, k}) .
		\end{multline*}
	\end{subequations}
	Using the latter to bound the second summand of the right hand side in the expression \eqref{Eq Variance Heights 1D Basic Equality}, we get
	\begin{equation*}
	\begin{split}
	\sum\limits_{k \neq i} \sum\limits_{\substack{\ell \neq i\\\ell \neq k}} \Exp(\X_{i, k}  \X_{i, \ell}) & =
	\sum\limits_{k \neq i} \sum\limits_{\substack{\ell \neq i\\\ell \neq k}} \Exp(\X_{i, k}) \Exp(\X_{i, \ell}) \\
	& =
	\sum\limits_{k \neq i} \Exp(\X_{i, k}) \sum\limits_{\substack{\ell \neq i\\\ell \neq k}}  \Exp(\X_{i, \ell}) \\
	& \leq
	\sum\limits_{k \neq i} \Exp(\X_{i, k}) \sum\limits_{\ell \neq i }  \Exp(\X_{i, \ell}) =
	\Exp^{2}(\uph_{i})  .
	\end{split}
	\end{equation*}
	Combining the equality of the second line above with \eqref{Eq Variance Heights 1D Basic Equality}, \textsc{Equation} \eqref{Eq Variance Heights 1D} follows. Finally, combining the inequality of the third line in the expression above with \eqref{Eq Variance Heights 1D Basic Equality}, the estimate \eqref{Ineq Variance Heights 1D} follows.
	%\qed
\end{proof}
Getting an estimate of the variance is useful to establish the number of Bernoulli trials (sampling) that have to be done in order to assure a confidence level for the numerical results. For instance, if the confidence interval is to furnish, respectively a 90\% and 95\% confidence, the number of trials is given by (see \cite{Thompson})
\begin{align*}%\label{Eq Number of Bernoullin Trials}
& \widetilde{K}_{90\%} \defining \big(\frac{1.645}{0.1}\big)^{2} \Var(\uph_{i}), &
& \widetilde{K}_{95\%} \defining \big(\frac{1.96}{0.05}\big)^{2} \Var(\uph_{i}) .
%\leq \big(\frac{1.96}{0.05}\big)^{2} \max\limits_{j\, = \, 1}^{N}\Var(\uph_{j}) .
\end{align*}
Therefore, we would like the value of $ \max\limits_{j \, = \, 1}^{N}\Var(\uph_{j}) $. However, it is not possible to give a closed formula, hence we aim for an estimate. For a fixed number of $ N $ points distributed inside a fixed interval, namely $ (0,1) $, it is well-known that the variance of the heights will be maximum when the points are equidistant i.e., the chances for an interval to be chosen attain its maximum level of uncertainty. Consequently, we adopt the maximum possible variance of a monotone partition $ \P $ whose endpoints are $ x_{i} = \dfrac{i}{N} $, $ i = 0, 1, \ldots, N $. We use the equality \eqref{Eq Variance Heights 1D Basic Equality} to compute numerically such maxima, the table \ref{Tb Maximum Variance Table} displays certain important values. 
\begin{table}[h!]
	\begin{center}
		%		\small{
		%		\rowcolors{2}{gray!25}{white}
		%\begin{minipage}{0.4\textwidth}
		\begin{tabular}{ c c c }
			\hline
			%			\rowcolor{gray!50}
			Exponent & Intervals & Maximum \\
			%			\rowcolor{gray!50}
			$ j $ & $ 3^{j} $ & Variance \\
			\hline
			%				$ 3^{1}$ 
			1 & 3 &	0.25 \\
			%				$ 3^{2}$ 
			2 & 9 &	1.32 \\
			%				$ 3^{3}$ 
			3 & 27 & 3.22 \\
			%				$ 3^{4}$ 
			4 & 81 & 5.32 \\
			%				$ 3^{5}$ 
			5 & 243	& 7.48 \\
			%				$ 3^{6}$ 
			6 & 729	& 9.67 \\
			%				$ 3^{7}$ 
			7 & 2187 & 11.86 \\
			\hline
		\end{tabular}
		%		}
	\end{center}
	\caption{ Maximum Variance Table }\label{Tb Maximum Variance Table}
\end{table}

An elementary linear regression adjustment gives
\begin{align*}%\label{Eq Variance Exponential Model}
& \Var\big(3^{j} \text{ congruent intervals}\big) = 1.99 j - 2.38 , &
& \kappa = 0.9967, &
& \sigma = 0.076 .
\end{align*}
Here $ \kappa $ is the correlation coefficient  and $ \sigma $ is the standard error. A quick change of variable gives
\begin{align}\label{Eq Variance Logarithmic Model}
& \Var\big(n \text{ congruent intervals} \big) = \frac{1.99}{\log 3} \log n - 2.38 , &
& \kappa = 0.9967, &
& \sigma = 0.076 ,
\end{align}
where $ n $ is the number of congruent intervals in the monotone partition $ \P $.

For a general problem, we compute the corresponding value $ \tilde{ \sigma}^{2}_{N} $ from the expression \eqref{Eq Variance Heights 1D Basic Equality}, adopt it from a table such as \ref{Tb Maximum Variance Table}, or use a regression model such as \textsc{Equation} \eqref{Eq Variance Logarithmic Model}. In the following, we use the adopted value of $ \tilde{ \sigma}^{2}_{N} $ in \eqref{Eq Estimate Number of Bernoullin Trials}, to compute the number of necessary Bernoulli trials $ K $ according to the desired confidence level
\begin{align}\label{Eq Estimate Number of Bernoullin Trials}
& K_{90\%} \defining \big(\frac{1.645}{0.1}\big)^{2} \tilde{ \sigma}^{2}_{N}, &
& K_{95\%} \defining \big(\frac{1.96}{0.05}\big)^{2} \tilde{ \sigma}^{2}_{N}.
%\leq \big(\frac{1.96}{0.05}\big)^{2} \max\limits_{j\, = \, 1}^{N}\Var(\uph_{j}) .
\end{align}
%
%\begin{minipage}{0.5\textwidth}
%%
%%\begin{table}[h!]
%	%
%	\begin{center}
%		\small{
%			\rowcolors{2}{gray!25}{white}
%			%\begin{minipage}{0.4\textwidth}
%%			\caption{ Maximum Variance Table }\label{Tb Maximum Variance Table}
%			\begin{tabular}{ c c c }
%				\hline
%				\rowcolor{gray!50}
%				Exponent & Intervals & Maximum \\
%				\rowcolor{gray!50}
%				 j & $ 3^{j} $ & Variance \\
%				\hline
%%				$ 3^{1}$ 
%				1 & 3 &	0.25 \\
%%				$ 3^{2}$ 
%				2 & 9 &	1.32 \\
%%				$ 3^{3}$ 
%				3 & 27 & 3.22 \\
%%				$ 3^{4}$ 
%				4 & 81 & 5.32 \\
%%				$ 3^{5}$ 
%				5 & 243	& 7.48 \\
%%				$ 3^{6}$ 
%				6 & 729	& 9.67 \\
%%				$ 3^{7}$ 
%				7 & 2187 & 11.86 \\
%				\hline
%			\end{tabular}
%		}
%	\end{center}
%%\end{table}
%%
%\end{minipage}
%%
%\begin{minipage}{0.5\textwidth}
%	\centering
%	{\includegraphics[scale = 0.4]{LogarithmicRegression.pdf}}
%\end{minipage}
%%
%
%
%%%%%%%%%%%%%%%%%%%%%%%%%%%%%%%%%%%%%%%%%%%%
\subsection{The Relationship Between Monotone Random Trees and iTrees}\label{Sec MRT and iTree}
%%%%%%%%%%%%%%%%%%%%%%%%%%%%%%%%%%%%%%%%%%%%
%
%
In the present section, we illustrate the link between the monotone random tree algorithm introduced in \textsc{Definition} \ref{Def Monotone Random Tree} and the iTree introduced in \textsc{Definition} \ref{Def Isolation Random Tree} for the 1D setting. To that end we first recall a definition and a proposition from basic graph theory (see \cite{GrossYellen})
\begin{definition}\label{Def Line Graph}
	The\textbf{ line graph} $ L(G) $ of a graph $ G $ has a vertex for each edge of $ G $, and two vertices in $ L (G) $ are adjacent if and only if the corresponding edges in $ G $ have a vertex in common.
\end{definition}
\begin{proposition}\label{Thm Line Graph}
	The line graph of a tree is also a tree. Moreover, $ h\big(L(T)\big) = h(T) - 1 $, where $ h(\cdot ) $ denotes the height of the graph.
\end{proposition}
\begin{proof}
	See \cite{GrossYellen}.
\end{proof}
In order to illustrate the relationship between monotone random trees and iTrees consider the tree $ T $  of \textsc{Figure} \ref{Fig Example MRF Tree} and transform it into the one displayed in \textsc{Figure} \ref{Fig Example MRF to iTree} (a), denoted by $ \G(T) $. The set of data $ S $, is given by the extremes of the intervals in $ \P $, each node hosting an interval has two children and the edges were labeled, using the corresponding left and right subsets generated when the interval is chosen. Abstract vertices were added whenever the label had a singleton. The root of $ \G(T) $ is also an abstract vertex; additionally, we introduce an edge connecting $ \troot (\G(T)) $ with $ \troot(T) $, labeled by the full set $ S $. Once the $ \G(T) $ tree is constructed, it is direct to see that its line graph $ L\big(\G(T)\big) $ is an isolation tree (iTree) of the data $ S $. 
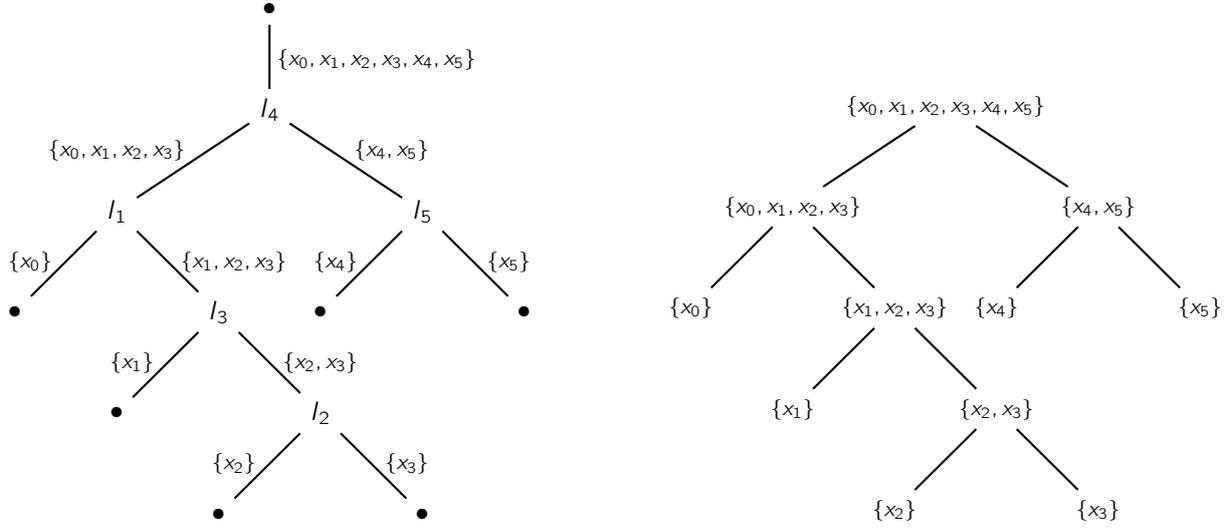
\begin{figure}[h]
	\centering
	\begin{subfigure}[Structure $ \G(T)$, extension and relabeling of the monotone random tree $ T $ of \textsc{Figure} \ref{Fig Example MRF Tree}.]
		{
			\begin{tikzpicture}
			[scale=.670,auto=left,every node/.style={}]
			\node (n4) at (5,8)  {$ I_{4} $};
			\node (n1) at (2,6)  {$ I_{1} $};
			\node (n3) at (4,4)  {$ I_{3} $};
			\node (n2) at (6,2)  {$ I_{2} $};
			\node (n5) at (8,6)  {$ I_{5} $};
			\node (n6) at (0,4)  {$ \bullet $};
			\node (n7) at (2,2)  {$  \bullet $};
			\node (n8) at (8,0)  {$ \bullet $};
			\node (n9) at (4,0)  {$ \bullet $};
			\node (n10) at (5,10)  {$ \bullet $};
			\node (n11) at (10,4)  {$ \bullet $};
			\node (n12) at (6,4)  {$ \bullet $};
			
			\node (n13) at (6.8,9)  {\footnotesize $ \{x_{0}, x_{1}, x_{2}, x_{3}, x_{4}, x_{5} \} $};
			\node (n14) at (2,7.2)  {\footnotesize $ \{x_{0}, x_{1}, x_{2}, x_{3} \} $};
			\node (n15) at (7.4,7.2)  {\footnotesize $ \{ x_{4}, x_{5} \} $};
			\node (n16) at (0.3,5)  {\footnotesize $ \{x_{0}\} $};
			\node (n17) at (4.2,5)  {\footnotesize $ \{x_{1}, x_{2}, x_{3}\} $};
			\node (n18) at (6.4,5)  {\footnotesize $ \{x_{4} \} $};
			\node (n19) at (9.7,5)  {\footnotesize $ \{x_{5} \} $};
			\node (n20) at (2.3,3)  {\footnotesize $ \{x_{1}\} $};
			\node (n21) at (6,3)  {\footnotesize $ \{x_{2}, x_{3}\} $};	
			\node (n22) at (4.3,1)  {\footnotesize $ \{x_{2}\} $};
			\node (n23) at (7.7,1)  {\footnotesize $ \{ x_{3}\} $};		
			
			%			\node(n18) at (6,7){ $\stackrel{ \blacksquare }{\begin{pmatrix}
			%					p_{k} \\ q_{\ell}
			%					\end{pmatrix} \in S  }$};
			%			\node(n19) at (3,7){$\stackrel{ \blacklozenge }{\begin{pmatrix}
			%					p_{2} \\ q_{\ell}
			%					\end{pmatrix} \in S  }$};
			%			\node(n20) at (6,3){$\stackrel{ \blacklozenge }{\begin{pmatrix}
			%					p_{k} \\ q_{1}
			%					\end{pmatrix} \in S  }$};		
			
			\foreach \from/\to in {n4/n1, n1/n3, n3/n2, n4/n5,
				n1/n6, n3/n7, n2/n8, n2/n9, n10/n4, n5/n11, n5/n12}
			\draw[line width = 1.0] (\from) -- (\to);  
			%			
			%			\foreach \from/\to  in{n18/n16, n18/n12, n6/n14, n7/n19, n19/n15,
			%				n2/n20, n20/n10, n3/n11}
			%			\draw[dashed] (\from) -- (\to);

			\end{tikzpicture}
		}
	\end{subfigure}
	\hspace{1cm}
	\begin{subfigure}[Structure $ L\big(\G(T)\big) $, line graph of $ \G(T) $.]
		{
			\begin{tikzpicture}
			[scale=.670,auto=left,every node/.style={}]
			%		\node (n4) at (5,8)  {$ I_{4} $};
			%		\node (n1) at (2,6)  {$ I_{1} $};
			%		\node (n3) at (4,4)  {$ I_{3} $};
			%		\node (n2) at (6,2)  {$ I_{2} $};
			%		\node (n5) at (8,6)  {$ I_{5} $};
			%		\node (n6) at (0,4)  {$ \bullet $};
			%		\node (n7) at (2,2)  {$  \bullet $};
			%		\node (n8) at (8,0)  {$ \bullet $};
			%		\node (n9) at (4,0)  {$ \bullet $};
			%		\node (n10) at (5,10)  {$ \bullet $};
			%		\node (n11) at (10,4)  {$ \bullet $};
			%		\node (n12) at (6,4)  {$ \bullet $};
			
			\node (n13) at (5,9)  {\footnotesize $ \{x_{0}, x_{1}, x_{2}, x_{3}, x_{4}, x_{5} \} $};
			\node (n14) at (2,7)  {\footnotesize $ \{x_{0}, x_{1}, x_{2}, x_{3} \} $};
			\node (n15) at (8,7)  {\footnotesize $ \{ x_{4}, x_{5} \} $};
			\node (n16) at (0, 5)  {\footnotesize $ \{x_{0}\} $};
			\node (n17) at (4, 5)  {\footnotesize $ \{x_{1}, x_{2}, x_{3}\} $};
			\node (n18) at (6,5)  {\footnotesize $ \{x_{4} \} $};
			\node (n19) at (10,5)  {\footnotesize $ \{x_{5} \} $};
			\node (n20) at (2,3)  {\footnotesize $ \{x_{1}\} $};
			\node (n21) at (6,3)  {\footnotesize $ \{x_{2}, x_{3}\} $};	
			\node (n22) at (4,1)  {\footnotesize $ \{x_{2}\} $};
			\node (n23) at (8,1)  {\footnotesize $ \{ x_{3}\} $};		
			
			%			\node(n18) at (6,7){ $\stackrel{ \blacksquare }{\begin{pmatrix}
			%					p_{k} \\ q_{\ell}
			%					\end{pmatrix} \in S  }$};
			%			\node(n19) at (3,7){$\stackrel{ \blacklozenge }{\begin{pmatrix}
			%					p_{2} \\ q_{\ell}
			%					\end{pmatrix} \in S  }$};
			%			\node(n20) at (6,3){$\stackrel{ \blacklozenge }{\begin{pmatrix}
			%					p_{k} \\ q_{1}
			%					\end{pmatrix} \in S  }$};		
			
			\foreach \from/\to in {n13/n14, n13/n15, n14/n16, n14/n17,
				n15/n18, n15/n19, n17/n20, n17/n21, n21/n22, n21/n23}
			\draw[line width = 1.0] (\from) -- (\to);  
			%			
			%			\foreach \from/\to  in{n18/n16, n18/n12, n6/n14, n7/n19, n19/n15,
			%				n2/n20, n20/n10, n3/n11}
			%			\draw[dashed] (\from) -- (\to);

			\end{tikzpicture}
		}
	\end{subfigure}
	\caption{Schematics of the bijection between $ \mrf $ and $ \irf $. Starting with the monotone random tree $ T $ of \textsc{Figure} \ref{Fig Example MRF Tree}, Figure (a) shows the first part of the transformation, while Figure (b) depicts the graph mapped in $ \irf $. 
	}
	\label{Fig Example MRF to iTree}
\end{figure}
\begin{remark}\label{Rem Monotone to iTree}
	It is possible to furnish a mathematically rigorous algorithm that would give a probability-preserving bijection between the spaces $ \mrf $ and $ \irf $ in 1D. This would deliver relationships between expected heights and topological properties, as well as properties of the variance for the  $ \irf $ in 1D setting. However, such construction is highly technical and contributes little to our topic of interest, therefore we omit it here. In contrast, we present  \textsc{Theorem} \ref{Thm Expectated Heights 1D Data} as a simple theoretical tool relating expected heights and topological properties.
\end{remark}
\begin{theorem}\label{Thm Expectated Heights 1D Data}
	Let $ S \defining \{x_{0}, x_{1}, \ldots, x_{N} \} $ be an arbitrary set of points on the line such that $ x_{0} < x_{1} < \ldots < x_{N} $. Let $ \uph_{i}: \irf \rightarrow \N $ be the random variable with $ \uph_{i}(T) $ defined as the depth of the data $ x_{i} $ in the isolation random tree $ T $. Then
	\begin{equation}
	\Exp(\uph_{i}) = \O\Big( \log \Big(\frac{x_{N} - x_{0}}{\dist(x_{i}, S - \{x_{i}\} )} \Big)\Big) .
	\end{equation}
\end{theorem}
\begin{proof}
	Define the intervals $ I_{n} \defining x_{n} - x_{n - 1}$ for every $ n \in [N] $ and the partition $ \P \defining \{ I_{n}: n\in \N \} $. Denote by $ \widetilde{\uph}_{i} $ the random variable indicating the depth in $ \mrf $ of the interval $ I_{i} $, then the following relations are direct
	\begin{align*}
	& \uph_{i} = 1 + \max\{ \widetilde{\uph}_{i}, \widetilde{\uph}_{i - 1} \} , & 
	& i \in \{ 1, \ldots, N\} ,\\
	& \uph_{0} = 1+ \widetilde{\uph}_{0}, &
	&  \uph_{N} = 1 + \widetilde{\uph}_{N} .
	\end{align*}
	%
	%From the ob servations above and observing that $ \dist(x_{i}, S -\{x_{i}\} ) = I_{i}$ for $ i \in \{0, N\} $, the result is direct due to \textsc{Theorem} \ref{Thm Expectated Heights 1D}. 
	%Let $ i \in \{ 1, \ldots, N\} $ be fixed, 
	Since $ \widetilde{\uph}_{j} \geq 0$ for all $ j\in [N] $, then $ \uph_{i} \leq  1 + \widetilde{\uph}_{i} + \widetilde{\uph}_{i - 1} $, for all $ i \in [N] $; hence 
	\begin{equation*}
	\begin{split}
	\Exp(\uph_{i}) 
	& \leq 1 + \Exp(\widetilde{\uph}_{i}) +
	\Exp(\widetilde{\uph}_{i - 1}) \\
	&\leq 
	1 +
	\O\Big(\log \Big( \dfrac{\sum\{\vert I_{i}\vert: n \in [N]\}}{\vert I_{i}\vert} \Big) \Big) 
	+ 
	\O\Big(\log \Big( \dfrac{\sum\{\vert I_{i}\vert: n \in [N]\}}{\vert I_{i-1}\vert} \Big) \Big) \\
	& \leq 
	\O\Big(\log \Big( \dfrac{\sum\{\vert I_{i}\vert: n \in [N]\}}{\min\{\vert I_{i-1}\vert, \vert I_{i}\vert \}} \Big) \Big).
	\end{split}
	\end{equation*}
	Since $ x_{N} - x_{0} = \sum\{\vert I_{n}\vert: n \in [N]\} $ and $ \dist(x_{i}, S-\{x_{i}\}) = \min\{\vert I_{i-1}\vert, \vert I_{i}\vert \} $, the proof is complete.
	%\qed
\end{proof}
%
%
%%%%%%%%%%%%%%%%%%%%%%%%%%%%%%%%%%%%%%%%%%%%
\subsection{A Modification of the IRF Method}\label{Sec Correction of the IRF Method}
%%%%%%%%%%%%%%%%%%%%%%%%%%%%%%%%%%%%%%%%%%%%%%%%%%%%%%%%%%%%%%%%%%%%%%%%%%%%%%%%%%%%%%%%%%
%
%
%
In the current section we present an alternative version of the IRF Method motivated by the mathematical understanding we have on the 1D setting. The Directional Isolation Random Forest Method (DIRF Method) works as follows
%
%
%\begin{definition}[DIRF Method]\label{Def Directional Isolation Random Tree}
%	%
%	Input: data set $ S \defining \{\x_{0}, \x_{1}, \ldots, \x_{N} \} \subseteq  \R^{d} $, number of Bernoulli trials $ K $. 
%	%
%	\begin{enumerate}[(i)]
%		%
%		\item Find the principal directions of the set $ S $.
%		
%		\item Project the data on each direction, i.e., generate  $ S_{i} \defining \{ \x \cdot \eversor_{i}: \x \in S\} $, for $ i = 1, \ldots, d $.
%		
%		\item For each $ S_{i} $ perform iTree $ K $-times and store the heights in separate files, one for each direction. Denote these files by $ \text{Log}_{i} $, $ i = 1, \ldots, d $.
%		
%		\item Compute the average height of $ \x\cdot \eversor_{i} $ for all $ (\x, i) \in S\times[d] $. Denote this averages by $ \{\upm(\x\cdot \eversor_{i}) : x \in S, i \in [d]\} $.
%		
%		\item For each $ x \in S $, define $ \uph_{\text{DIRF}}(x) \defining \min \{  \upm(\x\cdot \eversor_{i}): i = 1, \ldots, d\} $.
%		%
%	\end{enumerate}
%	%
%\end{definition}
\begin{definition}[The DIRF Method]\label{Def Directional Isolation Random Tree}
	Given an input data set $ S \defining \{\x_{0}, \x_{1}, \ldots, \x_{N} \} \subseteq  \R^{d} $, a number of Bernoulli trials $ K $ and an anomaly threshold criterion.  
	\begin{enumerate}[(i)]
		\item Find the principal directions of the set $ S $.
		
		\item Project the data on each direction, i.e., generate  $ \pi_{j}(S) \defining \{ \x \cdot \eversor_{j}: \x \in S\} $, for $ j = 1, \ldots, d $.
		
		\item For each Bernoulli trial, select at random one direction, namely $ j \in [d] $. Perform iTree (see \textsc{Defintion} \ref{Def Isolation Random Tree} and/or \textsc{Algorithm} \ref{Alg iTree Algorithm}) on $ S_{j} = \pi_{j}(S) $ and store the heights $ \{h(\x): \x\in S\} $ in a Log.
		
		%		\item For each $ S_{i} $ perform iTree $ K $-times and store the heights in separate files, one for each direction. Denote these files by $ \text{Log}_{i} $, $ i = 1, \ldots, d $.
		
		%		\item Compute the average height of $ \x $ for all $ (\x, i) \in S $. Denote this averages by $ \{\upm(\x\cdot \eversor_{i}) : x \in S, i \in [d]\} $.
		
		\item For each $ \x \in S $, define $ \uph_{\text{DIRF}}(\x) $ as the average height of the collection of heights $ \{h_{k}(\x): k = 1, \ldots, K\} $.
		
		\item Declare as anomalies 
		$ A \defining \{\x \in S: \uph_{\text{DIRF}}(\x) \text{ satisfies the anomaly threshold criterion}  \} $.
	\end{enumerate}
\end{definition}
It is understood that the number of Bernoulli trials $ K $ (see \textsc{Section} \ref{Sec Variance MRF}), is chosen to assure a confidence level for the computed value of the expected heights. Notice that
\begin{align}\label{Stmt convergence of the DIRF method}
& \uph_{\text{DIRF}}(\x) \xrightarrow[K \, \rightarrow \,\infty]{}
\frac{1}{d} \sum\limits_{i \, =\, 1}^{d} \Exp\big(\uph^{(i)}(\x)\big), &
& \text{for all } \x \in S .
\end{align}
% $ \big \{\upm(\x\cdot \eversor_{i}) : x \in S, i \in [d] \big\} $.
%
Here, $ \Exp\big(\uph^{(i)}(\x)\big) $ indicates the expected height of the data $ \x\cdot \eversor_{i} $ within the IRF of the set $ S_{i} = \pi_{i}(S) $, for all $ i \in [d] $.  The statement \eqref{Stmt convergence of the DIRF method} can be easily seen as follows: define $ A_{i}\defining \{k \in [K]: \text{ trial } k \text{ choses directon } i \} $, then
\begin{equation}
\uph_{\text{DIRF}}(\x) = \frac{1}{K}\sum\limits_{k \, = \, 1}^{K} h_{k}(\x) =
\sum\limits_{i \, = \, 1}^{d} \frac{ \# A_{i} }{K}\, \frac{1}{\# A_{i} }
\sum\limits_{k\, \in \,A_{i}} h_{k}(\x) .
\end{equation}
Due to the Law of Large Numbers (see \cite{BillingsleyProb}) it is clear that for all $ i \in [d]$, it holds that $ \frac{\# A_{i}}{K} \xrightarrow[K \, \rightarrow \,\infty]{} \frac{1}{d} $ and due to  \textsc{Corollary}  \ref{Thm the IRF Method Converges} $  \frac{1}{\# A_{i} }
\sum\limits_{k\, \in \,A_{i}} h_{k}(\x)  \xrightarrow[K \, \rightarrow \,\infty]{} \Exp \big(\uph^{(i)} \big( \x \big) \big) $.
\begin{theorem}[Computational Complexity of the Methods]\label{Thm Computational Costs}
	\begin{enumerate}[(i)]
		\item 	The computational cost of the DIRF method  is $ \mathcal{O}(N \log N) $.
		
		\item Under the hypothesis that DIRF and IRF have variance of the same order, the IRF method's computational cost is $ \mathcal{O}(N \log N) $. (As suggested in \cite{LiuIRF, Liu2012IsolationBasedAD}.)
	\end{enumerate}
\end{theorem}
\begin{proof}
	\begin{enumerate}[(i)]
		\item Combining \eqref{Eq Estimate Number of Bernoullin Trials} and \eqref{Eq Variance Logarithmic Model} with \textsc{Proposition} \ref{Thm Well posedness of iTree} the result follows.
		
		\item The hypothesis on the IRF variance implies that \eqref{Eq Estimate Number of Bernoullin Trials} and \eqref{Eq Variance Logarithmic Model} are valid. Therefore, the previous reasoning applies and the proof is complete.
	\end{enumerate} 
\end{proof}
In the table \ref{Tb IRF vs DIRF Schemes} below we present a qualitative comparison between the IRF method and the proposed DIRF method.
\begin{table}[h!]
	\begin{center}
		%		\small{
		%		\rowcolors{2}{gray!25}{white}
		%\begin{minipage}{0.4\textwidth}
		\begin{tabular}{ c c c }
			\hline
			%			\rowcolor{gray!50}
			Feature & IRF & DIRF \\
			\hline
			Mathematical Justification & Partial & Full \\
			Probabilistic Space Induced by the Method & Known & Known \\
			Convergence of the Method & Known & Known \\
			Necessary steps to isolate all data & $ N $ & $ N $ \\
			Number of Trials for Confidence Interval & Partially Known	& Known \\			
			Computational Complexity & $ \mathcal{O}(N \log N) $ &	$ \mathcal{O}(N \log N) $ \\
			Direction of Data Separation for an iTree realization& Variable within iTree	& Fixed within iTree \\
			Principal Components Analysis & Not necessary & Necessary \\
			\hline
		\end{tabular}
		%		}
	\end{center}
	\caption{IRF vs. DIRF Methods. It is understood that both methods are acting on the same set $ S = \{ \x_{1}, \ldots, \x_{N} \} \subseteq \R^{d} $, satisfying Hypothesis \ref{Hyp Non Repeated Coordinates}.}
	\label{Tb IRF vs DIRF Schemes}
\end{table}
%
%
%
%
%%%%%%%%%%%%%%%%%%%%%%%%%%%%%%%%%%%%%%%%%%%%%%%%%%%%%%%%%%%%%%%%%%%%%%%%%%%%%%%%%% 
\section{Numerical Experiments.}\label{Sec Numerical Experiments}
%%%%%%%%%%%%%%%%%%%%%%%%%%%%%%%%%%%%%%%%%%%%%%%%%%%%%%%%%%%%%%%%%%%%%%%%%%%%%%%%%%
%
%
%
The present section is devoted to the design and execution of numerical experiments in order to compare the performance of both methods: IRF and DIRF. The following aspects are important in this respect
\begin{enumerate}[(i)]
	\item The codes are implemented in python, some of the used libraries are pandas, scipy, numpy and matplotlib. 
	
	\item Although the experiments use benchmarks already labeled, we also use the distance-based definition of outlier, introduced in \cite{AngiulliPizzuti}:
	\begin{definition}\label{Def Outlier}
		Let $ r > 0 $ and $ 0 \leq  p \leq 1  $ be two fixed parameters and $ S \subseteq \R^{d} $ be a set. A point $ \x \in S $ is said to be an outlier with respect to the parameters $ r $ and $ p $ if
		\begin{equation}\label{Eq Outlier}
		\frac{\# \big( B (\x, r ) \cap S \big) }{ \# S } \leq p .
		\end{equation}
		Here $ B(\x, r) \defining \big\{\z \in \R^{d}: \Vert \x - \z \Vert \leq r  \big\} $, with $ \Vert \cdot \Vert $ the Euclidean norm.
	\end{definition}
	\item The number of sampled trees $ K $, inside the random forest (or Bernoulli trials) is computed combining \eqref{Eq Estimate Number of Bernoullin Trials} and \eqref{Eq Variance Logarithmic Model}.
	
	\item It is not our intention to debate the definition of an anomaly classifying threshold here. Therefore, our analysis runs through several quantiles acting as \textbf{anomaly threshold criteria}, which we adopt empirically based on observations of each case/example. That is, we state that $ \x $ is an anomaly if for the IRF method, the averaged height $ \uph_{\text{IRF}}(\x) $ belongs to the lowest 1\%, 2\%, 3\% (and so forth) of the set $ \{ \uph_{\text{IRF}}(\y) : \y \in S \} $. The analogous definition holds for the DIRF method.
	%) (respectively $ \uph_{\text{DIRF}}(\x) $ ) (respectively, the set $ \{ \uph_{\text{DIRF}}(\y) : \y \in S \} $).
	
	\item In both examples we use the PCA Method (Principal Components Analysis, see \cite{BishopDataScience}) because of the high dimension of the original data. On one hand, PCA is part of the DIRF method and on the other hand the IRF method works better when PCA is applied, hence we use it for both methods in the experiments in order to make them comparable. A different set of principal components is used in each example. These sets were chosen empirically according to the eigenvalues' order of magnitude, for better illustration of both methods (IRF and DIRF). Moreover, beyond the higher number of components both methods severely deteriorate due to the noise introduced by the lower order components.
	
	\item Our study will analyze, not only anomalies correctly detected but also the performance of the method against false positives. In practice, both methods IRF and DIRF need a threshold, under which all the values are declared anomalies by the method. Such procedure will include a number of false positives which we also quantify in our examples.
\end{enumerate}
%
%%
%%
%\begin{figure}[h!]
%	%	
%	\centering
%	%
%	\begin{subfigure}[IRF Method, Anomaly Detection. Original Labeling.]
%		{
%			\includegraphics[scale = 0.45]{IRF_Labeled_Quantiles_Breast.jpg}
%		}
%	\end{subfigure}
%	%
%	\hspace{1cm}
%	%
%	\begin{subfigure}[IRF Method, Anomaly Detection. Distance Based Artificial Labeling.]
%		{
%			\includegraphics[scale = 0.45]{IRF_Dist_Quantiles_Breast.jpg}
%		}
%	\end{subfigure}
%	%
%	\newline
%	%
%	\begin{subfigure}[DIRF Method, Anomaly Detection. Distance Based Artificial Labeling.]
%		{
%			\includegraphics[scale = 0.45]{PM_Labeled_Quantiles_Breast.jpg}
%		}
%	\end{subfigure}
%	%
%	\hspace{1cm}
%	%
%	\begin{subfigure}[DIRF Method, Anomaly Detection. Distance Based Artificial Labeling.]
%		{
%			\includegraphics[scale = 0.45]{PM_Dist_Quantiles_Breast.jpg}
%		}
%	\end{subfigure}
%	%
%	\caption{Anomaly detection percentages for Example \ref{Exm Benchmark}. All the graphics have the number of principal components in the $ x $-axis and multiple curves for the quantiles to be used as a threshold. Figure (a) and (c) depict anomalies detected by both methods with the original labeling, while figures (b) and (d) display the same concept for distance-based labels introduced in \textsc{Definition} \ref{Def Outlier}. 
%	}
%	\label{Fig IRF-DIRF Anomaly Detection Example Breast Results}
%\end{figure}
%%
%%
%
%
\begin{figure}[h!]
	\centering
	\begin{subfigure}[IRF Method. Number of Components vs Percentage of Anomaly Detection. Original Labeling.]
		{
			\includegraphics[scale = 0.48]{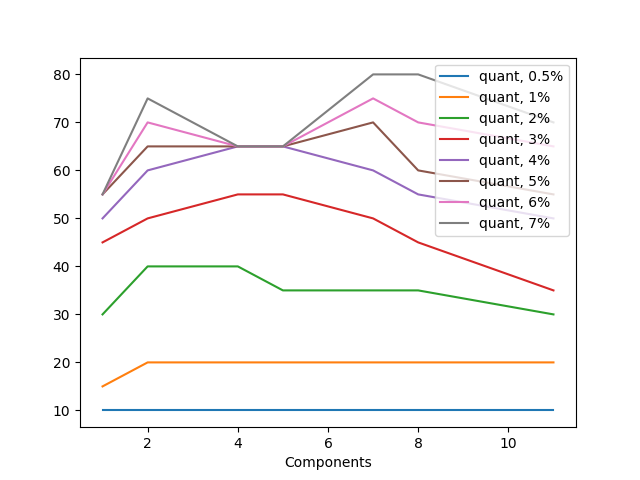}
		}
	\end{subfigure}
	%
	%	\hspace{1cm}
	%
	\begin{subfigure}[IRF Method. Number of Components vs Percentage of False Positives. Original Labeling.]
		{
			\includegraphics[scale = 0.48]{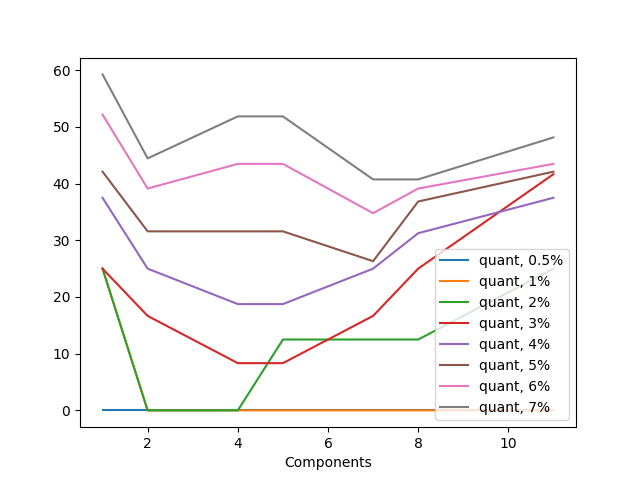}
		}
	\end{subfigure}
	\begin{subfigure}[DIRF Method. Number of Components vs Percentage of Anomaly Detection. Original Labeling.]
		{
			\includegraphics[scale = 0.48]{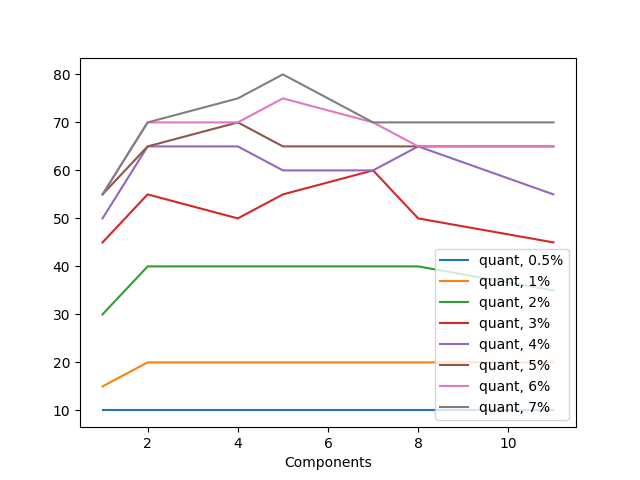}
		}
	\end{subfigure}
	%
	%	\hspace{1cm}
	%
	\begin{subfigure}[DIRF Method. Number of Components vs Percentage of False Positives. Original Labeling.]
		{
			\includegraphics[scale = 0.48]{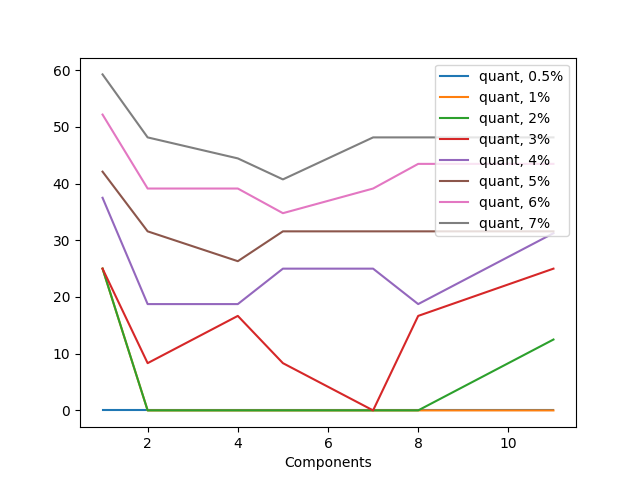}
		}
	\end{subfigure}
	\caption{Anomaly detection percentages for Example \ref{Exm Benchmark}, Breast Cancer Diagnosis. All the graphics have the number of principal components in the $ x $-axis and multiple curves for the quantiles to be used as a threshold. Figure (a) and (c) depict anomalies detected by both methods with the original labeling, while figures (b) and (d) display false positives introduced by both methods with the original labeling.
	}
	\label{Fig IRF-DIRF Labeled Example Breast Results}
\end{figure}
\begin{example}\label{Exm Benchmark}
	The first example uses the benchmark ``Breast Cancer Wisconsin (Diagnosis) Data Set", downloaded from \url{https://www.kaggle.com/uciml/breast-cancer-wisconsin-data}. Although the original data base contains 569 individuals, 213 patients (37.2\%) were diagnosed with cancer. It is clear that the patients diagnosed with cancer can not be considered anomalies if the full data base is used for the analysis. Therefore the original data set was modified: the subset of healthy patients was left intact and 20 randomly chosen patients with cancer (3.5\%) were chosen to complete the set. 
	
	Two labels were used, the diagnosis label coming from the original data set itself and a distance-based label computed according to \textsc{Definition} \ref{Def Outlier} with parameters $ r = 350 $, $ p = 0.05  $.  The number of sampled trees (Bernoulli trials) is given by $ K = 2250 $. The original dataset contains 32 columns, therefore we combine our technique with the PCA method (Principal Components Analysis); in this particular example we choose the 1, 2, 4, 5, 7, 8 and 11 first components. Our experiments show that both methods severely decay their quality from 11 components on, in particular both perform really poorly with the 32 components to be considered a viable option. Finally, the anomaly detection threshold quantiles are 0.5, 1, 2, 3, 4, 5, 6 and 7, which were chosen from observing the behavior of this particular case.	
	%	displayed in the table \ref{Tb PCA Table} below. 
	%Hence, we analyze the problem with its 1, 2, 4, 5, 7, 8 and 11 first components. Moreover, our experiments show that both methods severely decay their quality from 11 components on (due to the noise introduced by the lower order components). In particular both perform really poorly with the 32 components to be considered a viable option. Finally, the anomaly detection threshold quantiles are 0.5, 1, 2, 3, 4, 5, 6 and 7; chosen from observing the behavior of this particular case.
	
	%	%
	%	%
	%	\begin{table}[h!]
	%		\caption{PCA Values}\label{Tb PCA Table}
	%		%
	%		\begin{center}
	%			%
	%			\rowcolors{2}{gray!25}{white}
	%			\begin{tabular}{|c|c|c|}
	%				\hline
	%				\rowcolor{gray!50}
	%				\textbf{Order of Magnitude} & \textbf{\# of Eigenvalues} &  \textbf{\# of Components}\\
	%				\hline 
	%				$  10^{5} $ & 1 & 1\\
	%				$  10^{3} $ & 1& 2\\
	%				$ 10^{2} $ & 2 & 4\\
	%				$ 10^{1} $ & 1 & 5\\
	%				$ 10^{0} $ & 2 & 7\\
	%				$ 10^{-1} $ & 1 & 8\\
	%				$ 10^{-3} $ & 3 & 11\\
	%				\hline
	%			\end{tabular}
	%			%
	%		\end{center}
	%		%
	%	\end{table}
	%	%
	%	%
	%
	%
	\begin{table}[h!]
		\scriptsize
		\begin{center}
			%
			%			\rowcolors{2}{gray!25}{white}
			%		\begin{tabular}{| c | c c | c c | c c | c c | c c | c c | c c | c c |}
			\begin{tabular}{c c c  c c  c c  c c  c c  c c  c c  c c }
				\hline
				%				\rowcolor{gray!50}
				%\diagbox{Components}{Quantile [\%] 
				quantile [\%] 
				&\multicolumn{2}{ c }{0.5}&
				\multicolumn{2}{ c }{1}& \multicolumn{2}{ c }{2} & \multicolumn{2}{ c }{3}&
				\multicolumn{2}{ c }{4}& \multicolumn{2}{ c }{5}& \multicolumn{2}{ c }{6}&
				\multicolumn{2}{ c }{7} \\
				%\hline
				%				\rowcolor{gray!50}
				components & A & F & A & F & A & F & A & F & A & F & A & F & A & F & A & F \\
				\hline
				1 & 0.0 & 0.0 & 0.0 & 0.0 & 0.0 & 0.0 & 0.0 & 0.0 & 0.0 & 0.0 & 0.0 & 0.0 & 0.0 & 0.0 & 0.0 & 0.0 \\
				2 & 0.0 & 0.0 & 0.0 & 0.0 & 0.0 & 0.0 & -5.0 & 8.3 & -5.0 & 6.3 & 0.0 & 0.0 & 0.0 & 0.0 & 5.0 & -3.7 \\
				4 & 0.0 & 0.0 & 0.0 & 0.0 & 0.0 & 0.0 & 5.0 & -8.3 & 0.0 & 0.0 &-5.0 & 5.3 &-5.0 & 4.3 & -10.0 & 7.4 \\
				5 & 0.0 & 0.0 & 0.0 & 0.0 & -5.0 & 12.5 & 0.0 & 0.0 & 5.0 & -6.3 & 0.0 & 0.0&-10.0 & 8.7 &-15.0 & 11.1 \\
				7 & 0.0 & 0.0 & 0.0 & 0.0 & -5.0 & 12.5 & -10.0 & 16.7 & 0.0 & 0.0 & 5.0 & -5.3 & 5.0 & -4.3 & 10.0 & -7.4 \\ 
				8 & 0.0 & 0.0 & 0.0 & 0.0 & -5.0 & 12.5 &-5.0& 8.3 &-10.0 & 12.5 &-5.0 & 5.3 & 5.0 & -4.3 & 10.0 & -7.4\\
				11 & 0.0 & 0.0 & 0.0 & 0.0 & -5.0 & 12.5 & -10.0 & 16.7 & -5.0 & 6.3 & -10.0 & 10.5 & 0.0 & 0.0 & 0.0 & 0.0 \\
				\hline
			\end{tabular}
		\end{center}
		\normalsize
		\caption{Table of differences IRF -- DIRF, Breast Cancer Diagnosis, \textsc{Example} \ref{Exm Benchmark}. All the values are the difference of percentages. The columns ``A" and ``F" stand for anomalies and false positives respectively.}
		\label{Tb IRF - DIRF Methods Anomalies and False Positives}
	\end{table}

	The table \ref{Tb IRF - DIRF Methods Anomalies and False Positives} reports the difference of achievements attained by both methods when subtracting the DIRF from IRF. The predominance of negative and positive values in the columns ``A" and ``F" of the table \ref{Tb IRF - DIRF Methods Anomalies and False Positives} respectively, shows that the DIRF method performs better than the IRF method. Specially in the detection of false positives where DIRF performs significantly better than IRF: the former method presents convex curves, while the latter shows concave (or pseudo-convex) curves (see Figure \ref{Fig IRF-DIRF Labeled Example Breast Results}). %The graphics in \textsc{Figure} \ref{Fig IRF-DIRF Labeled Example Breast Results} reassure this first conclusion.

	Observe that the use of the quantiles is ``dual" in the following sense. It is clear that all the curves tend to shift upwards when the quantile is amplified. This is good from the anomaly detection point of view but bad from the false positives inclusion point of view and it is hardly surprising: the larger the threshold, the more likely we are to detect more anomalies, but also the higher the price of including false positives. For our particular example using a quantile of 4\% seems to be the ``balanced choice".
	
	%	Finally, about the number of components coming from PCA to be used, in this particular example, it seems clear that 5 is an inflection point for DIRF, while 7 is the inflection point for IRF. More specifically. DIRF shows superior performance below 5 components, between 6 and 7 IRF shows better performance and they both show poor and similar performance for more than 8 components.
	
	It must be observed that the quality of DIRF deteriorates with respect to IRF as we move along the diagonal of the table \ref{Tb IRF - DIRF Methods Anomalies and False Positives}, in particular DIRF performs poorly with respect to IRF from 7 PCA components and from the 6\%  quantile on.
	
	The same experiments were performed when using the artificial distance-based labeling introduced in \textsc{Definition} \ref{Def Outlier}.  It can be observed that both methods perform better for the anomaly detection, which is not unexpected because the DIRF method is strongly related to a distance function for anomalies, as shown in \textsc{Theorem} \ref{Thm Bound Quality}. However, both methods perform worse form the false positives inclusion point of view. Finally, the DIRF method performs better than the IRF method, although its superiority in the false positives inclusion is not as remarkable as in the first case.
\end{example}
\begin{figure}[h!]
	\centering
	\begin{subfigure}[IRF Method Number of Components vs Percentage of Anomaly Detection. Original Labeling.]
		{
			\includegraphics[scale = 0.48]{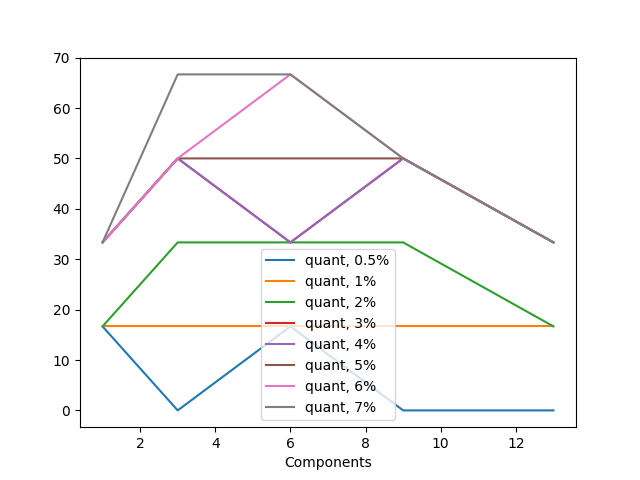}
		}
	\end{subfigure}
	%
	%	\hspace{1cm}
	%
	\begin{subfigure}[IRF Method. Number of Components vs Percentage of False Positives. Original Labeling.]
		{
			\includegraphics[scale = 0.48]{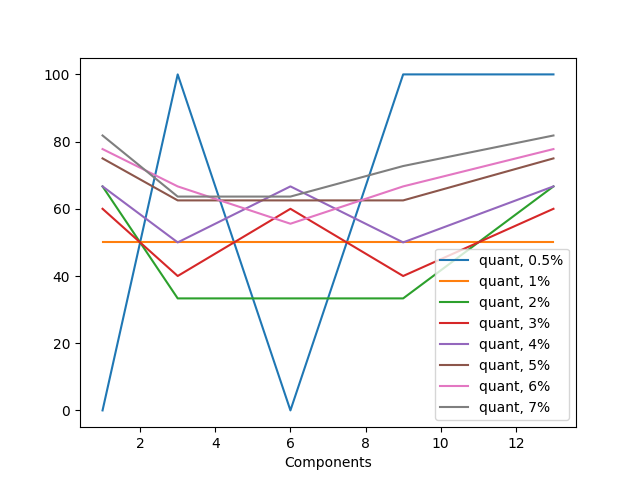}
		}
	\end{subfigure}
	\begin{subfigure}[DIRF Method. Number of Components vs Percentage of Anomaly Detection. Original Labeling.]
		{
			\includegraphics[scale = 0.48]{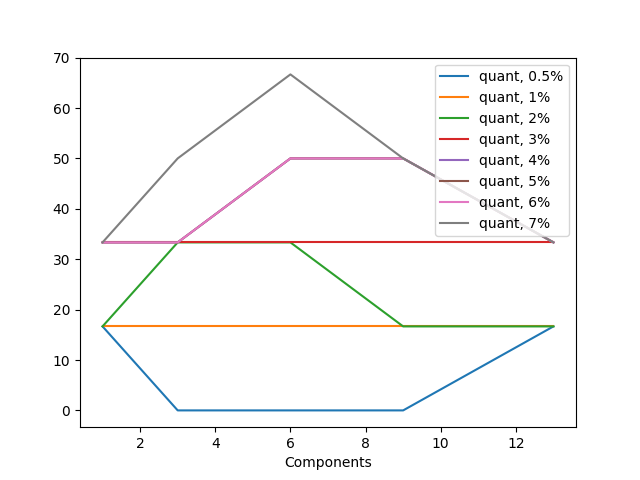}		}
	\end{subfigure}
	%
	%	\hspace{1cm}
	%
	\begin{subfigure}[DIRF Method. Number of Components vs Percentage of False Positives. Original Labeling.png]
		{
			\includegraphics[scale = 0.48]{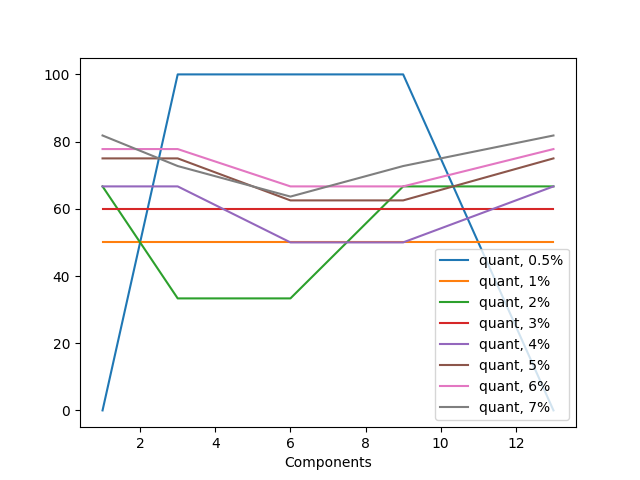}
		}
	\end{subfigure}
	\caption{Anomaly detection percentages for Example \ref{Exm Benchmark Lympho}, Lymphoma Diagnosis. All the graphics have the number of principal components in the $ x $-axis and multiple curves for the quantiles to be used as a threshold. Figure (a) and (c) depict anomalies detected by both methods with the original labeling, while figures (b) and (d) display false positives introduced by both methods with the original labeling.
	}
	\label{Fig IRF-DIRF Labeled Example Lympho Results}
\end{figure}
\begin{example}\label{Exm Benchmark Lympho}
	The second example uses a benchmark of lymphoma diagnosis, downloaded from \url{www.kaggle.com}. %\url{https://www.kaggle.com/uciml/breast-cancer-wisconsin-data}. 
	The dataset consists of 148 patients, with only 6 of them diagnosed having cancer, i.e. 4\%.
	
	Two labels were used, the diagnosis label coming from the original data set itself and a distance-based label computed according to \textsc{Definition} \ref{Def Outlier} with parameters $ r = 300 $, $ p = 0.05  $.  The number of sampled trees (Bernoulli trials) is given by $ K = 1800 $. The original dataset containes 18 columns, in contrast with \textsc{Example} \ref{Exm Benchmark}, the application of PCA yields eigenvalues whose order of magnitude does not change as abruptly. Therefore, we work with the 1, 3, 6, 9 and 13 first components. Our experiments show that none of the methods has a good performance for any number of components and its quality decays even more from 6 components on (due to the noise introduced by the lower order components). Finally, the anomaly detection threshold  quantiles are 0.5, 1, 2, 3, 4, 5, 6 and 7. These were chosen from observing the behavior of this particular case.

	\begin{table}[h!]
		\scriptsize
		\begin{center}
			%
			%						\rowcolors{2}{gray!25}{white}
			%			\begin{tabular}{| c | c c | c c | c c | c c | c c | c c | c c | c c |}
			\begin{tabular}{c  c c  c c  c c  c c  c c  c c  c c  c c }
				\hline
				%				\rowcolor{gray!50}
				%\diagbox{Components}{Quantile [\%] 
				quantile [\%] 
				&\multicolumn{2}{ c }{0.5} &
				\multicolumn{2}{ c }{1} & \multicolumn{2}{ c }{2} & \multicolumn{2}{ c }{3}&
				\multicolumn{2}{ c }{4}& \multicolumn{2}{ c }{5}& \multicolumn{2}{ c }{6}&
				\multicolumn{2}{ c }{7} \\
				%\hline
				%				\rowcolor{gray!50}
				components & A & F & A & F & A & F & A & F & A & F & A & F & A & F & A & F \\
				\hline
				1.0&0.0&0.0&0.0&0.0&0.0&0.0&0.0&0.0&0.0&0.0&0.0&0.0&0.0&0.0&0.0&0.0 \\
				3.0&0.0&0.0&0.0&0.0&0.0&0.0&16.7&-20.0&16.7&-16.7&16.7&-12.5&16.7&-11.1&16.7&-9.1 \\
				6.0&16.7&-100.0&0.0&0.0&0.0&0.0&0.0&0.0&-16.7&16.7&0.0&0.0&16.7&-11.1&0.0&0.0 \\
				9.0&0.0&0.0&0.0&0.0&16.7&-33.3&16.7&-20.0&0.0&0.0&0.0&0.0&0.0&0.0&0.0&0.0 \\
				13.0&-16.7&100.0&0.0&0.0&0.0&0.0&0.0&0.0&0.0&0.0&0.0&0.0&0.0&0.0&0.0&0.0 \\
				\hline
			\end{tabular}
		\end{center}
		\normalsize
		\caption{Table of differences IRF -- DIRF Lymphoma Diagnosis, \textsc{Example} \ref{Exm Benchmark Lympho}. All the values are the difference of percentages. The columns ``A" and ``F" stand for anomalies and false positives respectively.}\label{Tb IRF - DIRF Methods Anomalies and False Positives Lympho}
	\end{table}

	The table \ref{Tb IRF - DIRF Methods Anomalies and False Positives Lympho} reports the difference of achievements attained by both methods when subtracting the DIRF from IRF. Contrary to the previous example, there is a predominance of positive and negative values in the columns ``A" and ``F" of the table \ref{Tb IRF - DIRF Methods Anomalies and False Positives Lympho} respectively, showing that the IRF method performs better than the DIRF method with some few exceptions. This is also displayed in \textsc{Figure} \ref{Fig IRF-DIRF Labeled Example Lympho Results}. Nevertheless, we observe that the difference between methods is marginal. 
	
	As in the previous example, the 4\% quantile seems to be the ``balanced choice". In particular DIRF and IRF perform identically from 9 PCA components on and from the 6\%  quantile on.
	
	The same experiments were performed but using the artificial distance-based labeling introduced in \textsc{Definition} \ref{Def Outlier}. In this case, both methods perform almost identically and worse than in the case of the original labeling.
\end{example}
% 
%
%
%%
%\begin{example}\label{Exm Synthetically Generated}
%	The second example synthetically generated data using the normal distribution to generate 1000 data in 2D. We use the distance-based labeling introduced in \textsc{Definition} \ref{Def Outlier} with parameters 
%	
%	Normal $ r = 6 $, $ p = 0.01  $, $ 1.2\% $.  
%	Cauchy $ r = 200 $, $ p = 0.01 $, $ 1.3\%$.
%	Exponential $ r = 8 $, $ p = 0.01 $, $ 1.1\% $,
%	Uniform $ r = 0.58 $, $ p = 0.01 $, $1.4\%$.
%	
%	
%	The number of Bernoulli trials is given by $ K = 2750 $. In this particular case the PCA analysis is useful to set the data in the right directions, but unlike \textsc{Example} \ref{Exm Benchmark}, we do not exclude components when executing the experiments. 
%	
%	Finally, the quantiles are 0.5, 1, 2, 3, 4, 5, 6 and 7; chosen from observing the behavior of this particular case.
%\end{example}
%%
%
%%%%%%%%%%%%%%%%%%%%%%%%%%%%%%%%%%%%%%%%%%%%%%%%%%%%%%%%%%%%%%%%%%%%%%%%%%%%%%%%%%%%%%%%%%
%
%
%
%
%%%%%%%%%%%%%%%%%%%%%%%%%%%%%%%%%%%%%%%%%%%%%%%%%%%%%%%%%%%%%%%%%%%%%%%%%%%%%%%%%%
\section{Conclusions and Final Discussion.}\label{Sec Conclusion}
%%%%%%%%%%%%%%%%%%%%%%%%%%%%%%%%%%%%%%%%%%%%%%%%%%%%%%%%%%%%%%%%%%%%%%%%%%%%%%%%%%
%
%
The present work yields several conclusions listed below.
\begin{enumerate}[(i)]
	\item The IRF anomaly detection method introduced in \cite{LiuIRF, Liu2012IsolationBasedAD} has been \textbf{mathematically analyzed}. The \textbf{well-posednes}s (\textsc{Proposition} \ref{Thm Well posedness of iTree}) and the \textbf{convergence} of the algorithm (\textsc{Corollary} \ref{Thm the IRF Method Converges}) have been established as well as the \textbf{cardinality} (size) of the probabilistic space (\textsc{Theorem} \ref{Thm Cardinal of RIF}). 
	
	\item Under \textbf{mild assumptions}, it has been proved (see \textsc{Theorem} \ref{Thm Computational Costs}) that the \textbf{computational cost} of IRF method has order of magnitude $ \mathcal{O}(N \log N) $, as claimed in \cite{LiuIRF, Liu2012IsolationBasedAD}.
	
	\item It has been shown that although the IRF method is well-defined, convergent and its target values $ \{\Exp(\uph_{\x}): \x \in S\} $ are \textbf{inconclusive} when used as parameters for anomaly detection. (As shown in \textsc{Lemma} \ref{Thm Inconclusiveness in 2D} and \textsc{Theorem} \ref{Thm Inconclusiveness IRF Method} an outlier $ \x $ and a cluster point $ \y $ can get the same expected height $ \Exp(\uph_{x}) = \Exp(\uph_{y}) $.) Moreover, the IRF method can be deeply analyzed in the 1D case (\textsc{Theorem} \ref{Thm Expectated Heights 1D Data}) and it has been shown, from the theoretical point of view, that although it unquestionably detects outliers in this setting, its relationship with a notion of topological distance is \textbf{not certain} (see \textsc{Theorem} \ref{Thm Bound Quality}).
	
	\item Taking advantage of the \textbf{tractability} of IRF for the 1D case, we have given theoretical estimates of the variance (\textsc{Equation} \eqref{Eq Variance Logarithmic Model}) and derived the size of the sampling space (\textsc{Equation} \eqref{Eq Estimate Number of Bernoullin Trials}). This is necessary to guarantee confidence intervals for the empirically computed values of the expected heights. We have suggested a \textbf{modification} of the method in \textsc{Section} \ref{Sec Correction of the IRF Method}, named DIRF method (Directional Isolation Random Forest), whose differences with respect to the IRF method ar summarized in Table \ref{Tb IRF vs DIRF Schemes}.
\end{enumerate}
From the numerical examples in \textsc{Section} \ref{Sec Numerical Experiments}
\begin{enumerate}[(i)]
	\setcounter{enumi}{4}
	\item 
	%It is \textbf{unclear} whether or not the IRF or the DIRF method are \textbf{recommendable} for the outlier detection task. As pointed out in \textsc{Theorem} \ref{Thm Bound Quality} the relationship of the method with a notion of distance is not certain. 
	There is definitely \textbf{correlation} bewteen the heights computed by the methods (IRF and DIRF), but it could be strong as in \textsc{Example} \ref{Exm Benchmark} or weak as in \textsc{Example} \ref{Exm Benchmark Lympho}. 
	
	\item It is \textbf{clear} that the DIRF method is fully justified from the mathematical point of view which is desirable, however its relationship with a notion of distance is not certain (see the extent of \textsc{Theorem} \ref{Thm Bound Quality}) for multiple dimensions. In the examples we have seen IRF can perform better than DIRF, with a marginal difference. This is to be subject of extensive empirical evaluation in future work. 
	
	\item Both numerical examples \textbf{may suggest} that the adequate number of PCA components to introduce in the IRF and DIRF methods is \textbf{one third} of its total number of dimensions. Yet again, two experiments do not furnish enough numerical evidence to support such a conjecture, ergo this aspect needs to be further studied.
\end{enumerate}
As for future work,
\begin{enumerate}[(i)]
	\setcounter{enumi}{6}	
	\item Although we have proved that the IRF method is inconclusive as a means to classify anomalies, experience shows that it can provide satisfactory results in practice, it follows that the method is \textbf{correlated} with anomalies. There are two possible approaches for enhancing the method:
	\begin{enumerate}[a.]
		\item Find sufficient conditions for the data combinatorial configuration, to assure a quality certificate of IRF.
		
		\item Look for additional statistical parameters for anomaly detection (with computational cost no bigger than $ \mathcal{O}(n \log n) $) to complement/contrast the information furnished by IRF.
	\end{enumerate}
	Both research lines will be explored in future work.
	
	\item  The mathematical analysis the methods presented in \cite{LiuTingZhouSciForest} and \cite{HaririKindBrunner} will be pursued in future work, given that the performance of both is substantially superior to the original IRF. However, the probabilistic analysis will necessarily be different, because both references use random oblique hyperplanes for data separation, rather that axis-parallel as in IRF. This fact changes drastically the sampling space, from finite eligible directions in IRF, to uncountably many possible directions for the enhanced methods.
	
\end{enumerate}
%
%
%
%%%%%%%%%%%%%%%%%%%%%%%%%%%%%%%%%%%%%%%%%%%%%%%%%%%%%%%%%%%%%%%%%%%%%%%%%%%%%%%%%%%%%%%%%%
%%%%%%%%%%%%%%%%%%%%%%%%%%%%%%%%%%%%%%%%%%%%%%%%%%%%%%
\section*{Acknowledgements}
%\subsection{Acknowledgements} 
%An Acknowledgements section is started with \verb"\ack" or
%\verb"\acks" for \textit{Acknowledgement} or
%\textit{Acknowledgements}, respectively. It must be placed just
%before the References.
%\ack 
%\textit{Acknowledgements} 
The first Author wishes to thank Universidad Nacional de Colombia, Sede Medell\'in for supporting the production of this work through the project Hermes 54748 as well as granting access to Gauss Server, financed by ``Proyecto Plan 150x150 Fomento de la cultura de evaluaci\'on continua a trav\'es del apoyo a planes de mejoramiento de los programas curriculares" (\url{gauss.medellin.unal.edu.co}), where the numerical experiments were executed. Special thanks to Mr. Jorge Humberto Moreno C\'ordoba, our former student, who introduced us to the IRF method. 

The Authors wish to acknowledge the anonymous referees whose deep insight and kind suggestions, decisively enhanced the quality of this work.

%\subsection{Bibliography}
%%\nocite{*}% Show all bib entries - both cited and uncited; comment this line to view only cited bib entries;
%\bibliography{WileyNJD-AMA}%
%\begin{enumerate}[1]
%\item Use \verb"\bibliography{wileyNJD-AMA}" BST file for AMA reference style
%\item Use \verb"\bibliography{wileyNJD-APA}" BST file for APA reference style
%\item Use \verb"\bibliography{wileyNJD-AMS}" BST file for AMS reference style
%\item Use \verb"\bibliography{wileyNJD-VANCOUVER}" BST file for Vancouver reference style
%\item Use \verb"\bibliography{wileyNJD-ACS}" BST file for Chemistry reference style
%\end{enumerate}
%
%The normal commands for producing the reference list are:
%\bibliography{wileyNJD-AMA}
%\bibliography{bibliographie}

%%%%%%%%%%%%%%%%%%%%%%%%%%%%%%%%%%%%%%%%%%%%%%%%%%%%%%%%%%%%%%%%%%%%%%%%%%%%%%%%%%%%%%%%%%
\end{document}